\documentclass[a4paper,UKenglish,cleveref, autoref, thm-restate]{lipics-v2021}

\usepackage{graphicx,amsmath,xcolor} 
\usepackage{amsfonts}
\usepackage{amsmath}
\usepackage{amssymb}
\usepackage{amsthm}
\usepackage{stmaryrd}
\usepackage{mathtools}
\usepackage{nicefrac}
\usepackage[noabbrev]{cleveref} 
\usepackage{thmtools}
\usepackage{thm-restate} 
\usepackage{tikz}
\usetikzlibrary{shapes.geometric, positioning,decorations.pathreplacing, arrows, arrows.meta, fit}

\usepackage{amsmath,amssymb,stmaryrd,amsfonts,MnSymbol}

\usepackage{
color}
\usepackage{algorithmic}
\usepackage{graphicx}
\usepackage{textcomp}
\usepackage{mathtools}

\DeclareMathAlphabet{\mathbbold}{U}{bbold}{m}{n}

\usepackage[numbers]{natbib}

\usepackage{xifthen}
\usepackage{dsfont,amsmath,amssymb,mathtools,stmaryrd,bbm,color}

\DeclarePairedDelimiter{\ser}{\llbracket}{\rrbracket}

\allowdisplaybreaks

\numberwithin{equation}{section}

\newcommand{\ddefas}{\vcentcolon\vcentcolon=}

\newcommand{\free}{\mathrm{Free}}

\newcommand{\trop}{\mathrm{Trop}}
\newcommand{\arct}{\mathrm{Arct}}

\newcommand{\tuple}[1]{\ensuremath{\langle{#1}\rangle}}  

\newcommand{\Sr}{\mathcal{R}}
\newcommand{\sre}{r}

\newcommand{\Nz}{\mathbb{N}}
\newcommand{\Np}{\mathbb{N}_+}

\newcommand{\splus}{+}
\newcommand{\smal}{\cdot}

\newcommand{\szero}{e_{\oplus}}
\newcommand{\sone}{e_{\otimes}}
\newcommand{\bzero}{0}
\newcommand{\bone}{1}
\newcommand{\fplus}{\oplus}
\newcommand{\fmal}{\otimes}
\newcommand{\bigfplusop}{{\textstyle \bigoplus}} 
\newcommand{\bigfplus}{\bigoplus} 
\newcommand{\bigfmalop}{{\textstyle \bigotimes}} 
\newcommand{\bigfmal}{\bigotimes} 

\newcommand{\NPrm}{\ensuremath{\mathrm{NP}}}

\newcommand{\NPfinite}[1]{\NPrm\ensuremath{_{\mathrm{fin}}}(#1)}
\newcommand{\NPoldinf}[1]{\NPrm(#1)}
\newcommand{\NPnewinf}[1]{\NPrm\ensuremath{_{\infty}}(#1)}

\newcommand{\rel}[1]{\mathrm{Rel}_{#1}}
\newcommand{\ar}[1]{\mathrm{ar}_#1}

\newcommand{\inter}[1]{\mathcal{I}_{#1}}
\newcommand{\str}[1]{\mathrm{Str}(#1)}

\newcommand{\fo}[1]{\mathrm{FO}(#1)}
\newcommand{\wfo}[1]{\mathrm{wFO}(#1)}

\newcommand{\mso}[1]{\mathrm{SO}(#1)}
\newcommand{\wmso}[1]{\mathrm{wSO}(#1)}

\newcommand{\limrec}{\ensuremath{\mathit{LimDis}(\mathcal{R})}}
\newcommand{\karp}{\ensuremath{\mathit{Karp}}}

\newcommand{\rec}[1]{ \text{dis}_{#1}}

\newcommand{\dom}{\mathrm{dom}}

\newcommand{\B}{\mathbb{B}}

\newcommand{\R}{\mathbb{R}}

\newcommand{\mcp}{\mathcal{P}}

\newcommand{\mcv}{\mathcal{V}}

\newcommand{\mcx}{\mathcal{X}}

\newcommand{\mfa}{\mathfrak{A}}

\newcommand{\vecg}[1]{\bar{#1}} 

\newcommand{\nop}[1]{}

\newcommand{\myeqn}[3]{

\vspace*{#1}

\[#2
\]

\vspace*{#3} 

\noindent}
\newtheorem{define}[theorem]{Definition} 

\usepackage{todonotes}

\title{Fagin's Theorem for Semiring Turing Machines} 

\author{Guillermo Badia}{University of Queensland, Australia}{g.badia@uq.edu.au}{https://orcid.org/0000-0002-5597-6794}{}

\author{Manfred Droste}{University of Leipzig, Germany}{droste@informatik.uni-leipzig.de}{https://orcid.org/0000-0001-9128-8844}{}

\author{Thomas Eiter}{Technical University of Vienna, Austria}{eiter@kr.tuwien.ac.at}{https://orcid.org/0000-0001-6003-6345}{}

\author{Rafael Kiesel}{Technical University of Vienna, Austria}{rafael.kiesel@web.de}{https://orcid.org/0000-0002-8866-3452}{}

\author{Carles Noguera}{University of Siena, Italy}{carles.noguera@unisi.it}{https://orcid.org/0000-0003-4910-599X}{European Union's Marie Sklodowska--Curie grant no.\ 101007627 (MOSAIC project). GNSAGA, gruppo nazionale di ricerca INdAM.}

\author{Erik Paul}{University of Leipzig, Germany}{epaul@informatik.uni-leipzig.de}{https://orcid.org/0000-0002-0814-598X}{}

\authorrunning{G. Badia, M. Droste, T. Eiter, R. Kiesel,  C. Noguera, and E. Paul}

\Copyright{Guillermo Badia, Manfred Droste, Thomas Eiter, Rafael Kiesel, Carles Noguera and Erik Paul} 



\category{} 

\relatedversion{} 

\nolinenumbers 

\newcommand{\SumProd}[1]{\textsc{SumProd}(#1)}

\begin{document}

\maketitle

\begin{abstract} 
In recent years, quantitative complexity over semirings has been intensively investigated. In this context, Eiter and Kiesel 
(Semiring Reasoning Frameworks in AI and Their Computational Complexity, \emph{J. Artif. Intell. Res.}, 2023) introduced non-deterministic Turing Machines with semiring-weighted transitions (SRTMs) to capture the complexity of a manifold of semiring frameworks.
Beyond computational complexity, they posed the question of how we can relate the computational power of SRTMs to logical expressiveness. While this question was partially addressed for a more limited machine model by Badia et al.\ (Logical characterizations of weighted complexity classes, \emph{MFCS}, 2024), the full question remained open.

To answer it, we present an improved version of 
Eiter and Kiesel's SRTM model of computation. First and foremost, this enables us to prove a Fagin Theorem for the SRTM model, i.e., we show that the quantitative complexity class 
\NPnewinf{$\mathcal{R}$}, which comprises non-deterministic polynomial time computability in the improved SRTM model
over a commutative semiring $\mathcal{R}$, is captured by a version of weighted existential second-order logic that allows for predicates interpreted as semiring-annotated relations over $\mathcal{R}$. 
Furthermore, we argue that the new SRTM model is 
preferable over the original one and show that it reclaims some important results from~\cite{DBLP:journals/jair/EiterK23} that were flawed with respect to 
the latter.
\end{abstract}

\newpage
\setcounter{page}{1}
\section{Introduction}

{\bf Background and Motivation.} 
Fagin proved in a celebrated article~\cite{fagingeneralized}
that the class of $\mathsf{NP}$ languages coincides with the class of languages definable in existential second-order logic. This result essentially gave birth to the area of descriptive complexity, where one is interested in logical characterizations of complexity classes, a well-established topic today~\cite{DBLP:books/daglib/0095988}.

In this work, we go beyond decision problems and consider the descriptive complexity of quantitative models of computation. Here, an established approach is using the algebraic structure of \emph{semirings} to capture the quantitative aspects by annotating transitions in a computation with weights from a given semiring. This was fruitfully showcased most prominently for weighted automata~\cite{10.5555/1667106} and down the line for different Turing machine-style models, such as weighted Turing machines (originally called {\em algebraic Turing machines}~\cite{DBLP:journals/cc/DammHM02}) to naturally capture the complexity of various counting problems~\cite{DBLP:journals/iandc/Kostolanyi24,DBLP:conf/mfcs/BadiaDNP24}.

Specifically, we consider Semiring Turing Machines (SRTMs) introduce in~\cite{DBLP:journals/jair/EiterK23}. They differ from the more common weighted Turing machines in that semiring values appear explicitly on the tape and can be used as weights for transitions. This is necessary to capture, for example, the complexity of \emph{provenance} problems, such as the provenance queries of~\cite{GKV}, which require more than the finite number of fixed weights that can be provided in the weighted transition relations of weighted Turing machines.

A central problem that was left open in~\cite{DBLP:journals/jair/EiterK23} is a Fagin-style theorem in the SRTM context.\footnote{Barlag et al.\ \cite{DBLP:journals/corr/abs-2502-12939} asked for a Fagin-style result in the context of non-deterministic $\text{BSS}_K$ machines, which was answered in~\cite{barlag2025logicalapproachesnondeterministicpolynomial}. Similarly to our framework, their machines allow semiring inputs, but importantly they consider complexity classes of decision problems but not of function computations. Thus the second-order existential logic used in \cite{barlag2025logicalapproachesnondeterministicpolynomial} has  only Boolean second-order quantifiers (taking values from $\{0,1\}$), while our logic needs semiring-valued second-order quantifiers to obtain a Fagin-style theorem.}. Our main goal is to answer this question. However, we find that the SRTM model in \cite{DBLP:journals/jair/EiterK23} is not satisfactory for characterizing the quantitative problems it is concerned with. We thus amend the SRTM model and argue that the new definition satisfies properties desirable for such a model that were missed. As a byproduct, the amendment also fixes an edge case in the original model that falsified two important claims in~\cite{DBLP:journals/jair/EiterK23}.
%
To establish a Fagin-style theorem for the new SRTM model,
we introduce a weighted existential second-order logic by modifying the approach in~\cite{DBLP:journals/tcs/DrosteG07} in that semantically, we allow atomic formulas themselves to be assigned arbitrary semiring weights (rather than only $0$ or $1$),
similar as in semiring semantics for first-order logic~\cite{G-T} or, more restricted to lattices, in many-valued logic~\cite{Cintula-Hajek-Noguera:Handbook, doi:10.3233/FUN-2008-843-402}.

The logic we introduce in this paper can also be seen as a logical language 
for \emph{databases over semirings}, i.e., databases in which the tuples of the relations are annotated with elements from a fixed semiring. Relational databases are typically interpreted 
over the Boolean semiring $\mathbb{B}=(\{0,1\}, \vee, \wedge, 0,1)$, while \emph{bag}\/ semantics correspond to databases over the semiring $\mathbb{N}=(N, +, \cdot, 0,1)$ of the natural numbers. Interest in semirings beyond 
$\mathbb{B}$ and $\mathbb{N}$ was sparked by the influential work~\cite{GKV}, which introduced polynomial semirings as a framework for representing and analyzing the \emph{provenance} of database queries. Initially, 
the latter was applied to the positive fragment of relational algebra (specifically, unions of conjunctive queries) 
and later extended to encompass provenance in first-order logic~\cite{G-T} and least fixed-point logic~\cite{DBLP:conf/csl/DannertGNT21}. Further research expanded 
semiring semantics to a variety of database-related problems, 
including query containment~\cite{Green:2011,DBLP:journals/tods/KostylevRS14}, 
Datalog queries~\cite{DBLP:journals/jacm/KhamisNPSW24,DBLP:journals/pacmmod/ZhaoDKRT24}, database repairs~\cite{10.24963/kr.2025/32}, Codd's Theorem~\cite{10.1145/3767713}, and the relationship between local and global consistency in semiring-based relational structures~\cite{AtseriasK2023,DBLP:journals/pacmmod/AtseriasK24}. Another research line showed connections between semiring-based database models and quantum information theory~\cite{DBLP:conf/birthday/Abramsky13}.

\medskip 
\noindent
{\bf Our contribution.} The contributions in 
this paper are twofold, viz.\ conceptual and technical. On the conceptual side, we introduce a novel
(amended) {\em Semiring Turing Machine} model (SRTM, Defn~\ref{srtm}). We show that it is an improvement over the model of~\cite{DBLP:journals/jair/EiterK23} 
as it satisfies desirable properties that do not hold for the original model. As a byproduct, the new model allows us to properly establish the connections with classical counting problems claimed in~\cite{DBLP:journals/jair/EiterK23}. On the technical side, we give a logical characterization of the complexity class \NPnewinf{$\mathcal{R}$}, which generalizes non-deterministic polynomial time computability to our model, by a Fagin-style Theorem (Theorem~\ref{fagin}). To this end, we need a new weighted existential second-order logic (Defn~\ref{wesol}) designed 
for structures that might have $\mathcal{R}$-relations, that is, relations where the tuples are weighted with elements from a fixed semiring $\mathcal{R}$.

SRTMs have certain advantages over other computation models in the literature, as they allow users to provide semiring weights in the input and to perform computations with them. This is in contrast to weighted automata and weighted (resp.\ algebraic) Turing machines where one can only calculate with the weights of transitions. 
SRTMs relieve one from complex encoding of semiring values  (where the output grows exponentially) to represent such computations  
and allow to do them directly. A concrete example is computation over the rationals using the machine model of~\cite{DBLP:journals/iandc/Kostolanyi24}. 
As the machines of~\cite{DBLP:journals/iandc/Kostolanyi24} can access only finitely many weights and the rational numbers are, as a semiring, not finitely generated, there is \emph{no such} machine that can generate any rational number. 
SRTMs instead can use
semiring weights in the input for transition and, thus, do not have this limitation.
As a result, we are 
empowered to model many interesting and practically relevant problem frameworks, such as SumProduct problems over semirings~\cite{friesen2016sumproduct,bacchus2009solving} and probabilistic reasoning~\cite{de2007problog,kimmig2011algebraic}, while limiting the power of the machine to prevent explosive recursive computations.

\medskip 
\noindent
{\bf Related Work.} By a \emph{quantitative complexity class}, 
we mean a semiring generalization of a classical complexity class containing decision problems, where the objects of concern are \emph{power series} (mappings from words over some alphabet into a semiring). 
This notion subsumes classical counting classes 
and function classes like FP, FPSPACE, etc. Arenas et al.\ \cite{arenas19descriptive} introduced quantified second-order logic over the natural numbers to capture counting complexity classes such as \#P,  or SpanP, and function classes such as FP or FPSPACE by suitable fragments. However, a generic semiring setting was of no concern and this was only addressed in~\cite{DBLP:conf/mfcs/BadiaDNP24}, where logical characterizations (in the style of Fagin's Theorem) of quantitative versions of classical complexity classes were  studied. Importantly, Badia et al.'s computation model in \cite{DBLP:conf/mfcs/BadiaDNP24} is less general than ours as weighted Turing machines do not allow for semiring values in the input, and their results do not apply to our setting. Indeed, 
they acknowledged the problem we address as a challenging  open problem. 
Even more recently, Barlag et al.~\cite{DBLP:journals/corr/abs-2502-12939} discussed a further machine model, $K$-Turing deterministic machines (see \cref{appendix:comparision} for a comparison) and a Fagin-style theorem for the non-deterministic counterpart~\cite{barlag2025logicalapproachesnondeterministicpolynomial} with a  logic different from ours. 
Furthermore, there are many other quantitative models of computation such as weighted automata, Blum-Shub-Smale (BSS) machines, real counting with weights that can be approximated etc. SRTMs and the ensuing complexity classes are distinguished from them via the comparison in~\cite{DBLP:journals/jair/EiterK23}.


\section{Technical Preliminaries}
\label{sec:SRTM}
The basic building block to model the quantitative side of computation in our setting is given by semirings.

\begin{define}[Semiring] A \emph{semiring} is a  structure $\mathcal{R}=(R, \oplus, \otimes,0,1)$, where  addition $\oplus$ and multiplication $\otimes$ are binary operations on $R$, and  $0$, $1$ are elements of $R$ such that
\begin{itemize} 
 \item $(R, \oplus, 0)$  is a commutative monoid,  $(R, \otimes, 1)$ is a monoid, and $0\neq 1$;
 \item (distributivity): for all elements $a,b,c$ in $R$, we have that 
$a\otimes (b\oplus c) = a\otimes b \oplus a \otimes c$ and $ (b\oplus c) \otimes a = b\otimes a \oplus c \otimes a$;
\item (annihilation): for every element $a$ in $R$, we have that $a \otimes 0 = 0 \otimes a = 0$.
\end{itemize}
$\mathcal{R}$ is \emph{commutative}  if the monoid $(R,\otimes,1)$ is commutative.
We also use $e_\oplus, e_\otimes$ as alternative names for the elements $0, 1$, respectively.
\end{define}

Semirings allow us to abstract away the details of specific calculations and operations, while being on the one hand restricted enough by their axioms and on the other hand general enough to capture a vast range of quantitative problems. 

Some examples of semirings are the following:
 the \emph{Boolean semiring} $\B = (\{\bzero, \bone\}, \min, \max, \bzero, \bone)$,
 any bounded distributive lattice $(L, \lor, \land, 0, 1)$,
 the semiring of natural numbers $(\Nz, +, \cdot, 0, 1)$,
 the \emph{max-plus} or \emph{arctic semiring} $\arct = (\R_+ \cup \{-\infty\}, \max, +, -\infty, 0)$, where $\R_+$ denotes the set of non-negative real numbers,
 the \emph{min-plus} or \emph{tropical semiring} $\trop = (\R_+ \cup \{+\infty\}, \min, +, +\infty, 0)$.

\section{Semiring Turing Machines}
As we do not use Semiring Turing Machines (SRTM) of~\cite{DBLP:journals/jair/EiterK23} but give our own improved version, we begin by stating our expectations/desiderata towards SRTMs. Later, we will show that the our definition and that of~\cite{DBLP:journals/jair/EiterK23} differ in terms of satisfaction of the desiderata. 

First and foremost, SRTMs must be able to solve and capture the complexity of quantitative problems. Specifically, as~\cite{DBLP:journals/jair/EiterK23}, we are interested in problems like the stereotypical sum of products problem \SumProd{$\mathcal{R}$} over the semiring $\mathcal{R}$, which is to compute
\vspace{-0.2cm}
\begin{align}
\bigfplus_{X_1, \dots, X_m \in \mathcal{D}} \bigfmal_{i = 1}^{n} f_i(\vec{Y_i}), \label{eq:sumprod}
\end{align}
where $\mathcal{D}$ is a finite domain, $f_i: \mathcal{D}^{k_i} \rightarrow R, i = 1, \dots, n$ are functions with the set $R$ of semiring values as co-domain.

As such, SRTMs must be able to:
\begin{enumerate}
    \item compute semiring operations,
    \item dependent on logical symbols from the input, and
    \item using semiring values from the input.
\end{enumerate}
Here, 1. and 2. are clear. For 3., consider that $f_r : \{d\} \rightarrow R, d \mapsto r$, $\mathcal{D} = \{d\}$ defines a \SumProd{$\mathcal{R}$}-instance for each $r \in R$. Thus, there must a single SRTM, capable of returning all values in $R$, depending on the input. Since there are non-finitely generated semirings, i.e., semirings such that no finite set of values can represent all other values through combination via sum and product, this means giving the machine access to a fixed finite set of values is insufficient. At least, it must be able to work with the values from the input. 
Limiting ourselves to finitely generated semirings is not an option, as they do not include those used in practical applications such as for parameter learning for neuro-symbolic reasoning~\cite{NEURIPS2018_dc5d637e} or for probabilistic reasoning~\cite{de2007problog}.

Additionally, SRTMs must \emph{not} be able to
\begin{enumerate}
    \item distinguish the semiring values from the input, or
    \item reuse results of partial computations more than once.
\end{enumerate}
Here, 1. can be clearly observed in (\ref{eq:sumprod}): which values $f_i$ returns only depends on the variable assignments $X_i$ and not on the return value itself. 2. has more pragmatic concerns. It would allow the computation of recursive sequences such as $a_0 = 2, a_{i + 1} = a_i \cdot a_i = 2^{2^i}$, which goes beyond the exponential sums and polynomial products in (\ref{eq:sumprod}).  

With this in mind, we define our new machine model.
\begin{define}[SRTM]\label{srtm}
Given a  semiring $\mathcal{R}$, a \emph{Semiring Turing Machine} over $\mathcal{R}$ is $M=(R', Q, \Sigma_{I}, \Sigma, \iota, \mathcal{X}, \sqcup, \delta)$, where
\begin{itemize}
    \item $R' \subseteq R$ is a finite set of semiring values, (intuitively, this is a set of fixed values ``known'' to $M$),
    \item $Q$ is a finite set of states,
    \item $\Sigma_{I}$ is a finite set of symbols (the input alphabet),
    \item $\Sigma$ is a finite set of symbols (the tape alphabet) such that $\Sigma_{I} \subseteq \Sigma$,
    \item $\iota \in Q$ is the initial state,
    \item  $\sqcup \in \Sigma \setminus \Sigma_{I}$ is the blank symbol,
    \item $\mathcal{X} \in \Sigma \setminus \Sigma_{I}$ is the semiring placeholder,
    \item $\delta \subseteq \left(Q \times \Sigma \right)\times \left(Q \times \Sigma\right) \times \{-1,1\} \times (R' \cup \{\mathcal{X}\})$ is a weighted transition relation. 
\end{itemize}
An SRTM $M$ is \emph{terminating}, if there is an $\ell\colon (\Sigma_{I} \cup R)^* \longrightarrow \mathbb{N}$ such that for every $s \in (\Sigma_{I} \cup R)^*$ every computation path (defined as usual) of $M$ on $s$ has at most length $\ell(s)$.
\end{define}
Termination is necessary for us to ensure that the following notions are well defined, as it keeps our sums and products of semiring values finite.

A tuple $((q, \sigma), (q', \sigma'), d, r) \in \delta$ intuitively means that there is a transition from current state $q$ to new state $q'$, when reading symbol $\sigma$. Through the transition, we write $\sigma'$ and change the index of the head by $d$. The weight of the transition is $r$ if $r \in R'$, or the semiring value in the tape cell under the head if $r = \mathcal{X}$. Note for this, that tape cells always contain both a read-write symbol $\sigma$ and a read-only semiring value. We now formalize above intuition.

In the following, given $x \in (\Sigma \times R)^{\mathbb{N}}$ such that $x = (\sigma_i, r_i)_{i \in \mathbb{N}}$, we use $x^{\Sigma}_n = \sigma_n$ and $x^{R}_n = r_n$. 

\begin{define}[SRTM value, Transitions]
The \emph{value} $\mathrm{val}(c)$ of a terminating SRTM $M$ on a configuration $c = (q, x, n)$, where $q \in Q$ is a state, $x \in (\Sigma \times R)^{*}$ is the string on the tape, and $n \in \mathbb{N}$ is the head position, is recursively defined by
\[
\mathrm{val}(c) = \bigoplus_{((q_1, \sigma_1), (q_2, \sigma_2), d, r) \in \delta, q_1 = q, \sigma_1 = x^{\Sigma}_n} \mathrm{wt}(r, c) \otimes \mathrm{val}(q_2, x[x^{\Sigma}_n/\sigma_2], n + d),
\]
where $x[x^{\Sigma}_n/\sigma_2]$ denotes the tape string, where the alphabet value at index $n$ was replaced by $\sigma_2$ and $\mathrm{wt}(r, c)$ denotes the weight of the transition, which is $r$, if $r \in R'$, and $x^{R}_n$, if $r = \mathcal{X}$; the empty sum has value $e_{\otimes}$. 
\end{define}
Since problems typically have a mixture of alphabet letters and semiring values in the input string, we need to transform said input string into the form of our tape strings. For this, we write the alphabet letters (resp.\ semiring values) into the alphabet (semiring) part of the tape cells with the same index. Formally:
\begin{define}[Initial Configuration]\label{def:initial}
Given a string $s \in (\Sigma_{I} \cup R)^{*}$, the initial tape string $x(s) \in (\Sigma \times R)^{\mathbb{N}}$ is defined by 
\begin{itemize}
\item  $x(s)^{\Sigma}_i = s_i$, if $i < |s|$ and $s_i \in \Sigma_{I}$,
\item  $x(s)^{\Sigma}_i = \mathcal{X}$, if $i < |s|$ and $s_i \in R$, 
\item $x(s)^{\Sigma}_i = \sqcup$, if $i \geq |s|$,
\item  $x(s)^{R}_i = s_i$, if $i < |s|$ and $s_i \in R$,
\item 
    $x(s)^{R}_i = e_{\oplus}$, if ($i < |s|$ and $s_i \in \Sigma_{I}$) or $i \geq |s|$.
    \end{itemize}
\end{define}

\begin{define}[SRTM function/output, \NPnewinf{$\mathcal{R}$}]
\label{def:NPnewinf}
Let  $M$  be a terminating SRTM. The output function
$\lVert M\rVert: (\Sigma_I \cup R)^* \longrightarrow R$
is defined by  $\lVert M\rVert (s) = \mathrm{val}(\iota, x(s), 0) $
for each  $s \in (\Sigma_I \cup R)^*$. Then,  \NPnewinf{$\mathcal{R}$} is defined as the complexity class
of those functions $f\colon(\Sigma_I \cup R)^* \longrightarrow R$ such that $f = \lVert M\rVert $ for some SRTM $M$ for which $\ell(s)$  is polynomially bounded in $|s|$.

\end{define}

\begin{example}[Conditional product]
\label{ex:conditional-product}
Consider this function over an arbitrary semiring $\mathcal{R}$:
\vspace{-0.2cm}
\begin{align*}
    f \ \colon& \ (\Sigma \cup R)^* \longrightarrow R, 
    x \mapsto \left\{ \begin{array}{cr}
        s_1 \fmal \dots \fmal s_n & \text{ if } x  = \sigma_1 \dotsb \sigma_n s_1 \dotsb s_n,\text{ where } \sigma_i \in \Sigma, s_i \in R, \\
        e_{\oplus} & \text{ otherwise}
    \end{array}\right.
\end{align*}
that is, if $x$ starts with any number of alphabet symbols followed by the same number of semiring values, then $f(x)$ is their product; otherwise $f(x)$ yields zero, i.e., $e_\oplus$.

An SRTM that represents $f$ is given by $M = (R', Q, \Sigma_{I}, \Sigma, \iota, \mathcal{X}, \sqcup, \delta)$, where
\vspace{-0.2cm}
\begin{align*}
    R' &= \{e_{\oplus}, e_{\otimes}\}\\
    Q &= \{\iota, \mathrm{right}, \mathrm{left}, \mathrm{turn\_left}, \mathrm{turn\_right}, \mathrm{fin}\}\\
    \Sigma &= \Sigma_{I} \cup \{\sqcup, \mathcal{X}\}\\
    \delta &= \{\ ((\iota, \mathcal{X}), (\mathrm{fin}, \mathcal{X}), 1, e_{\oplus}), & \\
    & \phantom{= \;\;\{} ((\iota, \sigma), (\mathrm{right}, \sqcup), 1, e_{\otimes}), \ \  ((\iota, \sqcup), (\mathrm{fin}, \sqcup), 1, e_{\otimes}), & \\
    & \phantom{= \;\;\{} ((\mathrm{right}, \sigma), (\mathrm{right}, \sigma), 1, e_{\otimes}), \ \  ((\mathrm{right}, \mathcal{X}), (\mathrm{right}, \mathcal{X}), 1, e_{\otimes}), & \\
    & \phantom{= \;\;\{}  ((\mathrm{right}, \sqcup), (\mathrm{turn\_left}, \sqcup), -1, e_{\otimes}), & \\
    & \phantom{= \;\;\{} ((\mathrm{turn\_left}, \mathcal{X}), (\mathrm{left}, \sqcup), -1, \mathcal{X}), \ \  ((\mathrm{turn\_left}, \sigma), (\mathrm{fin}, \sigma), -1, e_{\oplus}), & \\
    & \phantom{= \;\;\{} ((\mathrm{turn\_left}, \sqcup), (\mathrm{fin}, \sqcup), -1, e_{\oplus}), & \\
    & \phantom{= \;\;\{} ((left, \mathcal{X}), (\mathrm{left}, \mathcal{X}), -1, e_{\otimes}), \ \  ((\mathrm{left}, \sigma), (\mathrm{left}, \sigma), -1, e_{\otimes}), & \\
    & \phantom{= \;\;\{} ((\mathrm{left}, \sqcup), (\mathrm{turn\_right}, \sqcup), 1, e_{\otimes}), & \\
    & \phantom{= \;\;\{} ((\mathrm{turn\_right}, \mathcal{X}), (\mathrm{fin}, \mathcal{X}), 1, e_{\oplus}), \ \   ((\mathrm{turn\_right}, \sigma), (\mathrm{right}, \sqcup), 1, e_{\otimes}), & \\
    & \phantom{= \;\;\{}  ((\mathrm{turn\_right}, \sqcup), (\mathrm{fin}, \sqcup), 1, e_{\otimes}) & \}.
\end{align*}
The idea is that we remove alphabet values (from the left) and semiring values (from the right) one by one in such a manner that the first removed value is an alphabet value and the last removed value is a semiring value. If this fails, we transition with weight zero, thus failing the computation. To achieve that $f(x)$ is the product of the semiring values, if the computation succeeds, we transition with the weight under the head, whenever we remove a semiring value. We work out the computation for $\Sigma_{I} = \{a\}$, $\mathcal{R} = \mathbb{N}$ and $x = a1$.
This gives us as initial configuration:
\myeqn{-15pt}{(\iota, \begin{array}{cc|cccc}
    \dots & \sqcup & a & \mathcal{X} & \sqcup & \dots \\
    \dots & 0 & 0 & 1 & 0 & \dots
\end{array}, 0)}
{-2pt}
Here, the table denotes the tape, as the alphabet part in the first row and the semiring part in the second row, and the vertical bar denotes the ``start'' of the tape, i.e., the zero-th tape cell is directly to the right; we remind that $e_\oplus = 0$ and $e_\otimes = 1$.

In the first step, we have to use the transition $((\iota, \sigma), (\mathrm{right}, \sqcup), 1, 1)$, with $\sigma = a$, resulting in aggregated weight $1$ (as $oplus =1$) and configuration
\myeqn{-15pt}{
(\mathrm{right}, \begin{array}{cc|cccc}
    \dots & \sqcup & \sqcup & \mathcal{X} & \sqcup & \dots \\
    \dots & 0 & 0 & 1 & 0 & \dots
\end{array}, 1)}
{-5pt}
Next, we have to use the transition $((\mathrm{right}, \mathcal{X}), (\mathrm{right}, \mathcal{X}), 1, e_{\otimes})$, resulting in accumulated weight $1$ and configuration
\myeqn{-15pt}{
(\mathrm{right}, \begin{array}{cc|cccc}
    \dots & \sqcup & \sqcup & \mathcal{X} & \sqcup & \dots \\
    \dots & 0 & 0 & 1 & 0 & \dots
\end{array}, 2)}{-5pt}
Next, we have to use the transition $((\mathrm{right}, \sqcup), (\mathrm{turn\_left}, \sqcup), -1, e_{\otimes})$, resulting in accumulated weight $1$ and configuration
\myeqn{-15pt}{
(\mathrm{turn\_left}, \begin{array}{cc|cccc}
    \dots & \sqcup & \sqcup & \mathcal{X} & \sqcup & \dots \\
    \dots & 0 & 0 & 1 & 0 & \dots
\end{array}, 1)}{-5pt}
Next, we have to use the transition $((\mathrm{turn\_left}, \mathcal{X}), (\mathrm{left}, \sqcup), -1, \mathcal{X})$, resulting in accumulated weight $1$ and configuration
\myeqn{-15pt}{
(\mathrm{left}, \begin{array}{cc|cccc}
    \dots & \sqcup & \sqcup & \sqcup & \sqcup & \dots \\
    \dots & 0 & 0 & 1 & 0 & \dots
\end{array}, 0)}{-5pt}
Next, we have to use the transition $((\mathrm{left}, \sqcup), (\mathrm{turn\_right}, \sqcup), 1, e_{\otimes})$, resulting in accumulated weight $1$ and configuration
\myeqn{-15pt}{
(\mathrm{turn\_right}, \begin{array}{cc|cccc}
    \dots & \sqcup & \sqcup & \sqcup & \sqcup & \dots \\
    \dots & 0 & 0 & 1 & 0 & \dots
\end{array}, 1)}{-5pt}
Next, we have to use the transition $((\mathrm{turn\_right}, \sqcup), (\mathrm{fin}, \sqcup), 1, e_{\otimes})$, resulting in accumulated weight $1$ and configuration
\myeqn{-15pt}{
(\mathrm{fin}, \begin{array}{cc|cccc}
    \dots & \sqcup & \sqcup & \sqcup & \sqcup & \dots \\
    \dots & 0 & 0 & 1 & 0 & \dots
\end{array}, 2)}{-5pt}
Since $\mathrm{fin}$ does not have any transitions, we halt with overall weight $1$, as expected.
\end{example}

As SRTMs differ from classical Turing machines in their ability to treat semiring values as first class citizens on the tape and in the input string, a direct comparison 
of them is not possible. Instead, for such a comparison, we always need to consider semirings together with an injective encoding function $f\colon R \longrightarrow \Sigma_{I}^*$ that maps semiring values to tape strings. 

Regarding other machine models that use semiring weights, 
if there are no semiring values on the tape, then we would have computations as for 
weighted Turing machines in~\cite{DBLP:journals/iandc/Kostolanyi24}.
\begin{lemma}
For every commutative semiring $\mathcal{R}$ and $f\colon \Sigma_{I}^* \longrightarrow R$ there exists an SRTM $M$ such that $\lVert M \rVert  = f$ iff there exists a weighted Turing machine $M'$ such that $\lVert M' \rVert = f$.
\end{lemma}
This can be easily verified by observing that the definitions of SRTMs and weighted Turing machines are equivalent, when no weights are allowed as inputs. It follows that over the Boolean semiring, when the input does not contain semiring values, SRTMs are equivalent to non-deterministic Turing machines.
\medskip 

\noindent{\bf Comparison with the SRTM model in \cite{DBLP:journals/jair/EiterK23}.}
Different from requiring 
$\delta \subseteq \left(Q \times \Sigma \right)\times \left(Q \times \Sigma\right) \times \{-1,1\} \times (R' \cup \{\mathcal{X}\})$, 
\cite{DBLP:journals/jair/EiterK23} required $\delta \subseteq \left(Q \times (\Sigma \cup R) \right)\times \left(Q \times (\Sigma \cup R)\right) \times \{-1,1\} \times R$ is a weighted transition relation, s.t.\ for each $((q_1, \sigma_1), (q_2, \sigma_2), d , r) \in \delta$ the following holds:
\begin{enumerate}
    \item if $\sigma_1 \in R$ or $\sigma_2 \in R$, then $\sigma_1 = \sigma_2$,
    \item $r \in R'$ or $r = \sigma_1 \in R$
    \item if $\sigma_1 \in R$, then
    \begin{enumerate}
        \item for all $\sigma_1' \in R$ 
    we have $((q_1, \sigma_1'),$ $(q_2, \sigma_1'), d, \sigma_1') \in \delta$ or
        \item for all $\sigma_1' \in R$ 
    we have $((q_1, \sigma_1'),$ $(q_2, \sigma_1'), d, r) \in \delta$.
    \end{enumerate}
\end{enumerate}
We see two differences here: (D1) In the old model of \cite{DBLP:journals/jair/EiterK23}, semiring values were sole 
contents of a tape cell rather than 
shared with alphabet symbols. 
(D2) Instead of using the placeholder $\mathcal{X}$, \cite{DBLP:journals/jair/EiterK23} used an infinite transition relation allowing, however, only (in 2.\ and 3.) to transition with a value from $R'$ or the value on the tape as a weight.

Both differences have serious consequences. 
Together with constraint 1., which disallows overwriting semiring values, (D1) 
effects that the old machine model cannot fully handle arbitrarily long chains of consecutive semiring values on the tape: as a machine 
has only finitely many states, it cannot navigate to the $k$-th semiring value in such a 
sequence for some 
large enough $k$, 
as states are bound to repeat   
and distinguishing the head position is impossible.
Consequently, the conditional product function in Example~\ref{ex:conditional-product} is not computable by the model of \cite{DBLP:journals/jair/EiterK23}, which is formally proved in~\Cref{appendix:full_proofs}.

As for (D2), 
consider in the new model 
tuples $t_1 = ((q, \mathcal{X}), (q', \mathcal{X}), d, r^*)$ and $t_2 = ((q, \mathcal{X}),$ $(q', \mathcal{X}), d, \mathcal{X})$ in $\delta$.
When we transition from $(q, \mathcal{X})$ to $(q', \mathcal{X})$ in direction $d$, we have two transitions, one by $t_1$ with weight $r^*$  and one by $t_2$ with the weight of the semiring value at the head position.
To express the same in the old model, we would need to have $((q, r), (q', r), d, r^*) \in \delta$ and $((q, r), (q', r), d, r) \in \delta$ for all $r \in R$. This would 
work as in the new model, except when $r = r^*$. In this case, both transitions would collapse into the single transition $((q, r^*), (q', r^*), d, r^*)$, which would only be executed once as $\delta$ is a set. This, however, means that the old model is in fact not oblivious to the semiring values on the tape.   

Both of these effects are clearly undesirable and not in the spirit of the definition. 
We thus consider it reasonable to prove a Fagin's Theorem for our SRTM model and not the original one in~\cite{DBLP:journals/jair/EiterK23}.
After this main result has been established, we will return to a more detailed analysis of the consequences and the relation between the two models in~\Cref{sec:comparison_old}

\section{A Fagin Theorem for \texorpdfstring{\NPnewinf{$\mathcal{R}$}}{NP(R)\_{new}}} 
\label{sec:fagin-thm} 

In this section, 
we introduce a weighted logic to characterize \NPnewinf{$\mathcal{R}$} and prove a Fagin-style theorem (Theorem~\ref{fagin}). Weighted logics, with weights coming from an arbitrary semiring, were developed originally to obtain a weighted version of the B\"uchi--Elgot--Trakhtenbrot Theorem, showing that a certain weighted monadic second-order logic has the same expressive power on words as weighted automata~\cite{DBLP:journals/tcs/DrosteG07}. The intuition behind the idea is that we incorporate the machinery of a semiring as part of the logical vocabulary, and thus we have a  multiplicative conjunction, an additive disjunction, and similarly, a multiplicative universal quantifier and an additive existential quantifier, as well as all elements of the semiring in question as logical constants (interpreted by themselves).


\medskip
\noindent
{\bf A suitable weighted existential second-order logic.} A  \emph{vocabulary} (or \emph{signature}) $\mu$ is a pair $(\rel{\mu} \cup \rel{\mu}^{\text{w}}, \ar{\mu})$ where $\rel{\mu}$ and $\rel{\mu}^{\text{w}}$ are disjoint sets of relation symbols and $\ar{\mu} \colon \rel{\mu} \cup \rel{\mu}^{\text{w}}\longrightarrow \Np$ is the arity function. Given a semiring $\mathcal{R}$, a \emph{$\mu$-structure}  $\mfa$ is a pair $(A, \inter{\mfa})$ where $A$ is a set, called the \emph{universe of\/ $\mfa$}, and $\inter{\mfa}$ is an \emph{interpretation}, which maps every symbol $S \in \rel{\mu}$ to a relation $S^\mfa \subseteq A^{\ar{\mu}(S)}$ 
and every symbol $W \in \rel{\mu}^{\text{w}}$ to an  $\mathcal{R}$-relation, i.e., a mapping $W^\mfa:  A^{\ar{\mu}(W)} \longrightarrow \mathcal{R}$. We assume that
structures are finite, i.e., $A$ is a finite set. 
A structure  $\mfa$ is 
\emph{ordered} if it is 
on a vocabulary $\mu \cup\{<\}$ where  $<$ is interpreted as a linear ordering of $A$ with endpoints. 
$\str{\mu}_<$ denotes the class of \emph{all finite ordered $\mu$-structures}.

We provide a countable set $\mcv$ of first and second-order variables, where lower case letters like $x$ and $y$ denote first-order variables and capital letters like $X$ and $Y$ denote second-order variables. Each second-order variable $X$ comes with an associated arity, denoted by $\mathrm{ar} (X)$.

In all of this section, let $\mu$  be a signature and $\Sr$ a semiring. Next, we define first-order formulas $\beta$ 
and weighted first-order formulas $\varphi$ over $\mu$ and $\Sr$
as follows.
\begin{define}[Syntax of $\mathrm{wFO}$]  \label{wfo} 
The syntax of $\mathrm{wFO}[\mathcal{R}]$ is given as follows:
\vspace{-0.2cm}
\begin{align*}
\beta &\ddefas  
 S(x_1, \ldots, x_n) \mid \lnot \beta \mid \beta \lor \beta \mid \exists x.\beta \\
\varphi &\ddefas \beta \mid \sre \mid  
 W(x_1, \ldots, x_m) \mid \varphi \fplus \varphi \mid \varphi \fmal \varphi \mid \bigfplusop x.\varphi \mid \bigfmalop x. \varphi 
\end{align*}
where $S \in \rel{\mu}$, and $n = \ar{\mu}(S)$; $W \in \rel{\mu}^{\text{w}}$,  and $m = \ar{\mu}(W)$; $x, x_1, \ldots, x_n \in \mcv$ and  $x, x_1, \ldots, x_m \in \mcv$ are first-order variables; and  $\sre \in \mathcal{R}$. 
\end{define}

Likewise, we define the formulas of weighted second-order logic as follows:

\begin{define}[Syntax of $\mathrm{wSO}$]  \label{wsol} 

The syntax of $\mathrm{wSO}[\mathcal{R}]$ is given
as follows:
\vspace{-0.2cm}
\begin{align*}
\beta &\ddefas  X(x_1, \ldots, x_n) \mid  
 S(x_1, \ldots, x_n) \mid \lnot \beta \mid \beta \lor \beta \mid \exists x.\beta \mid \exists X.\beta\\
\varphi &\ddefas \beta \mid \sre \mid  
 W(x_1, \ldots, x_n) \mid \varphi \fplus \varphi \mid \varphi \fmal \varphi \mid \bigfplusop x.\varphi \mid \bigfmalop x. \varphi \mid \bigfplusop X. \varphi \mid \bigfmalop X. \varphi,
\end{align*}
with $S \in \rel{\mu}$ and $n = \ar{\mu}(S)=\mathrm{ar} (X)$; $W \in \rel{\mu}^{\text{w}}$ and  $m = \ar{\mu}(W)$; $x, x_1, \ldots, x_n \in \mcv$  and  $x, x_1, \ldots, x_m \in \mcv$ are first-order variables;  $X \in \mcv$ is a second-order variable; and $\sre \in \mathcal{R}$. 
\end{define}
In the second layer of our syntax, even though we have new predicates  $ W(x_1, \ldots, x_n)$ that will be interpreted semantically as semiring-annotated relations, we only allow for second-order quantification over the variables of the initial layer that are meant to be interpreted as regular set-theoretic relations.
We also use the usual abbreviations $\land$, $\forall$, $\to$, and $\leftrightarrow$. By $\fo{\mu}$ and $\wfo{\mu,\Sr}$ we denote the sets of all first-order formulas over $\mu$ (they do not involve the symbols in $\rel{\mu}^{\text{w}}$) and all weighted first-order formulas over $\mu$ and $\Sr$, respectively, and by $\mso{\mu}$ and $\wmso{\mu,\Sr}$ we denote the sets of all second-order formulas over $\mu$ and all weighted second-order formulas over $\mu$ and $\Sr$, respectively.

The notion of \emph{free variables} is 
as usual, i.e., the operators $\exists, \forall, \bigfplusop$, and $\bigfmalop$ bind variables;
$\free(\varphi)$ is the set of all free variables of $\varphi$. A formula $\varphi$ with $\free(\varphi) = \emptyset$ is called a \emph{sentence}. For a tuple $\vecg{\varphi} = (\varphi_1, \ldots, \varphi_n) \in \wmso{\mu, \Sr}^n$, we define $\free(\vecg{\varphi}) = \bigcup_{i=1}^n \free(\varphi_i)$.

We define the semantics of $\mathrm{SO}$ and $\mathrm{wSO}$ as follows.
Let $\mfa = (A, \inter{\mfa})$ be a $\mu$-structure, and $\mcv$ be a set of first-order and second-order variables. 
A $(\mcv, \mfa)$-assignment $\rho$ is a function $\rho \colon \mcv \longrightarrow A \cup \bigcup_{n \geq 1}\mcp(A^n)$ such that, whenever $x \in \mcv$ is a first-order variable, we have $\rho(x) \in A$, and whenever $X \in \mcv$ is a second-order variable, we have $\rho(X) \subseteq A^{\mathrm{ar} (X)}$.
Let $\dom(\rho)$ be the domain of $\rho$. For a first-order variable $x \in \mcv$ and an element $a \in A$, the \emph{update} $\rho[x \mapsto a]$ is defined through $\dom(\rho[x \mapsto a]) = \dom(\rho) \cup \{x\}$, $\rho[x \mapsto a](\mcx) = \rho(\mcx)$ for all $\mcx \in \mcv\setminus\{x\}$, and $\rho[x \mapsto a](x) = a$. For a second-order variable $X \in \mcv$ and a set $I \subseteq A^{\mathrm{ar} (X)}$, the update $\rho[X \mapsto I]$ is defined in a similar fashion. By $\mfa_\mcv$ we denote the set of all $(\mcv, \mfa)$-assignments.

For $\rho \in \mfa_\mcv$ and $\beta \in \mso{\tau}$ the relation ``$(\mfa, \rho)$ satisfies $\beta$'' is defined as usual as
\begin{align*}
\
 (\mfa, \rho) \models S(x_1, \ldots, x_n)&& \Longleftrightarrow &\hspace{1em}  x_1, \ldots, x_n \in \dom(\rho) \text{ and } (\rho(x_1), \ldots, \rho(x_n)) \in S^\mfa, S \in \rel{\mu}\\
 (\mfa, \rho) \models X(x_1,..., x_n) && \Longleftrightarrow & \hspace{1em}x_1,..., x_n, X \in \dom(\rho) \text{, } \langle\rho(x_1), \dots,  \rho(x_n)\rangle \in \rho(X) \\
 (\mfa, \rho) \models \lnot \beta && \Longleftrightarrow & \hspace{1em}(\mfa, \rho) \models \beta \text{ does not hold}\\
 (\mfa, \rho) \models \beta_1 \lor \beta_2 && \Longleftrightarrow & \hspace{1em}(\mfa, \rho) \models \beta_1 \text{ or } (\mfa, \rho) \models \beta_2\\
 (\mfa, \rho) \models \exists x. \beta && \Longleftrightarrow & \hspace{1em}(\mfa, \rho[x \mapsto a]) \models \beta \text{ for some } a \in A\\
 (\mfa, \rho) \models \exists X. \beta && \Longleftrightarrow & \hspace{1em}(\mfa, \rho[X \mapsto I]) \models \beta \text{ for some } I \subseteq A.
\end{align*}

Let $\varphi \in \wmso{\mu,\Sr}$ and $\mfa \in \str{\mu}_<$. Let $a_1, \dots, a_k$ be an enumeration of $\mfa$ according to 
$<$ 
that serves as the interpretation of $<$, and for each integer $n$, let $I_1^n, \dots, I_{l_n}^n$  be an enumeration of the subsets of $A^n$ according to the lexicographic ordering induced by the interpretation of $<$.   The \emph{(weighted) semantics} of $\varphi$ is a mapping $\ser{\varphi}(\mfa,\cdot) \colon \mfa_\mcv \longrightarrow \Sr$ inductively defined as 
 \vspace{-0.3cm}
\begin{align*}
\ser{\beta}(\mfa, \rho)\quad &=&& \hspace{-3em} \begin{cases} \sone & \text{if } (\mfa, \rho) \models \beta \\ \szero & \text{otherwise} \end{cases} & \\
\ser{W(x_1, \dots, x_n)}(\mfa, \rho)\quad &=&& \hspace{-3em} W^\mfa(\rho(x_1), \dots, \rho(x_n)), \ \text{where} \ W \in \rel{\mu}^{\text{w}} & \\
\ser{\sre}(\mfa, \rho)\quad &=&& \hspace{-3em} \sre\\
\ser{\varphi_1 \fplus \varphi_2}(\mfa, \rho) \quad&=&& \hspace{-3em} \ser{\varphi_1}(\mfa, \rho) \oplus \ser{\varphi_2}(\mfa, \rho) & \\
\ser{\varphi_1 \fmal \varphi_2}(\mfa, \rho)\quad &=&& \hspace{-3em} \ser{\varphi_1}(\mfa, \rho) \otimes \ser{\varphi_2}(\mfa, \rho) & 
\nop{********
\\
\ser{\bigfplusop x. \varphi}(\mfa, \rho)\quad &=&& \hspace{-3em}\bigoplus_{a \in A} \ser{\varphi}(\mfa, \rho[x \mapsto a]) & \\
\ser{\bigfmalop x. \varphi}(\mfa, \rho) \quad&=&& \hspace{-3em}   \bigotimes_{1\leq i \leq k} \ser{\varphi}(\mfa, \rho[x \mapsto a_i]) & \\
\ser{\bigfplusop X. \varphi}(\mfa, \rho)\quad &=&& \hspace{-3em} \bigoplus_{I \subseteq A^{\mathrm{ar} (X)}} \ser{\varphi}(\mfa, \rho[X \mapsto I]) & \\
\ser{\bigfmalop X. \varphi}(\mfa, \rho)\quad &=&& \hspace{-3em}  \bigotimes_{1 \leq i \leq l_{\mathrm{ar} (X)}} \ser{\varphi}(\mfa, \rho[X \mapsto I_i^{\mathrm{ar} (X)}]).&
**********} 
\end{align*}

\vspace{-1.5\baselineskip}

\begin{align*}
\!\!\!\ser{\bigfplusop x. \varphi}(\mfa, \rho) \  &=\ \bigoplus_{a \in A} \ser{\varphi}(\mfa, \rho[x \mapsto a]) & \ser{\bigfplusop X. \varphi}(\mfa, \rho) \ &= \!\bigoplus_{I \subseteq A^{\mathrm{ar} (X)}} \ser{\varphi}(\mfa, \rho[X \mapsto I]) & \\
\!\!\!\ser{\bigfmalop x. \varphi}(\mfa, \rho)\ &= \ \bigotimes_{1\leq i \leq k} \ser{\varphi}(\mfa, \rho[x \mapsto a_i])  & \ser{\bigfmalop X. \varphi}(\mfa, \rho) \ &=  \! \bigotimes_{1 \leq i \leq l_{\mathrm{ar} (X)}} \ser{\varphi}(\mfa, \rho[X \mapsto I_i^{\mathrm{ar} (X)}]).  &
\end{align*}
Thanks to the lexicographic ordering, our product quantifiers have a well-defined semantics. Note that if the semiring is commutative, 
the semantics of the universal quantifiers
is defined by using any order for the factors in the products. 
Semantics for weighted logics where predicates are evaluated as power series already appears in~\cite{doi:10.3233/FUN-2008-843-402}.

\begin{define}[Weighted existential  second-order logic] \label{wesol}
Weighted existential second-order logic over the semiring $\mathcal{R}$, in symbols  $\mathrm{wESO}[\mathcal{R}]$, is  the fragment of $\mathrm{wSO}[\mathcal{R}]$ which includes all formulas of $\mathrm{wFO}[\mathcal{R}]$ (with the addition of formulas $X(x_1, \ldots, x_n)$, where $X$ is a relational variable, in the atomic case of the $\beta$s) and is closed under the quantifiers $\exists X$ and  $\bigfplusop X$.
\end{define}

Finally, it is worth pointing out that the first-order logic  (without negation) used  as query language in~\cite{Green:2011} 
for conjunctive queries over semiring-annotated databases is a fragment of the weighted logic  $\mathrm{wFO}$ introduced in this section. More precisely, we would simply have to consider $\mu$ with $\rel{\mu} =\emptyset$, and the following formulas with the same semantics as above:
\vspace{-0.2cm}
\begin{align*}
\varphi &\ddefas 
 W(x_1, \ldots, x_n) \mid \varphi \fplus \varphi \mid \varphi \fmal \varphi \mid \bigfplusop x.\varphi \mid \bigfmalop x. \varphi.
\end{align*}


\noindent
{\bf Capturing the quantitative complexity class \texorpdfstring{\NPnewinf{$\mathcal{R}$}}{NP(R)\_{new}}.}
Let  \( \mathfrak{A} = (A,<, S_1, \dotsc, S_k) \)  be  an ordered structure with \(A = \{a_1, \dotsc, a_n\}\). 
To encode $\mathfrak{A}$ as a string of semiring values,  we might first take \(\mathrm{enc}(S_i) \) for \(S_i \in \rel{\mu}^{\text{w}}\) to be the string of length $|A^{\mathrm{ar}(S_i)}|$ where the \(i\)th element is the weight of the \(i\)th  (according to the lexicographic ordering induced by $<$) $\mathrm{ar}(S_i)$-tuple of elements from $A$.
For \(S_i \in \rel{\mu}\) we proceed similarly, namely by defining the weight of each tuple as the letter \(0\)
if the tuple is not in \(S_i\), and as the letter \(1\) if the tuple is in \(S_i\).
Then
\(\mathrm{enc}(\mathfrak{A}) = 0^n1 \mathrm{enc}(S_1) \ \dotsb \ \mathrm{enc}(S_k)\),
where \(0\) and \(1\)
are letters.

Next we need to define what it means for the quantitative complexity class \NPnewinf{$\mathcal{R}$} to be characterized by a suitable weighted logic. 

\begin{define}\label{ch}
Consider a weighted logic $\mathrm{L}[\Sr]$ (with weights in a semiring $\Sr$). We say that $\mathrm{L}[\Sr]$ \emph{captures} \NPnewinf{$\mathcal{R}$} over ordered structures in the vocabulary $\mu=\{S_1, \dots, S_j\}$ if the following holds:
(1) 
for every $\mathrm{L}[\Sr]$-sentence $\phi$, some $P\in \NPnewinf{\mathcal{R}}$ exists s.t.\ $P(\text{enc}(\mathfrak{A})) = \ser{\phi}(\mathfrak{A})$ for every finite ordered $\mu$-structure $\mathfrak{A}$, and
(2) 
for every $P\in \NPnewinf{\mathcal{R}}$, some $\mathrm{L}[\Sr]$-sentence $\phi$ exists s.t.\ $P(\text{enc}(\mathfrak{A})) = \ser{\phi}(\mathfrak{A})$ for every finite ordered $\mu$-structure $\mathfrak{A}$.
\end{define}

With this definition at hand, we are finally able to state the main result of this paper, our weighted version of Fagin's theorem.
A proof is given in \cref{appendix:full_proofs}.

\begin{theorem}\label{fagin}
For each semiring $\mathcal{R}$, the logic $\mathrm{wESO}[\mathcal{R}]$ captures the quantitative complexity class \NPnewinf{$\mathcal{R}$} over ordered structures. 
\end{theorem}

\begin{remark}
    Theorem~\ref{fagin} also holds for finite structures without an ordering $<$ if the semiring $\mathcal{R}$ is idempotent and commutative.
    The first direction of the proof above does not reference the order at all, and for the second direction we can sum the very last formula over all possible orders and employ this ``guessed'' order in the proof.
    Due to idempotence and the fact that the weight of the last formula does not depend on the guessed order,
    as the SRTM does not have access to it,
    this is equivalent to guessing an arbitrary order for the proof.
\end{remark}




\section{On \texorpdfstring{\NPnewinf{$\mathcal{R}$}}{NP(R)\_{new}} and \texorpdfstring{\NPoldinf{$\mathcal{R}$}}{NP(R)\_{old}}} 
\label{sec:comparison_old}
Finally, we discuss what the change in the SRTM definition means for the results in~\cite{DBLP:journals/jair/EiterK23}, beginning with the precise relation of \NPnewinf{$\mathcal{R}$} and the corresponding complexity class \NPoldinf{$\mathcal{R}$}  of ~\cite{DBLP:journals/jair/EiterK23}.
For this, we formalize the  ability of the old model to distinguish semiring values:
\begin{define}[Limited Distinction Function]
Let $\mathcal{R}$ be a semiring and $S \subseteq R$ be a finite subset of its values. 
Then the 
 function $\rec{S}\colon R 
\longrightarrow R$ 
such that $\rec{S}(x) = \bigoplus_{r \in S} r$ if $x \,{\in S}\,$ and $\rec{S}(x) =$ $ x \oplus \bigoplus_{r \in S} r$, otherwise, is a \emph{limited distinction function}. Furthermore, we denote by  
$\limrec = \{\rec{S} \mid S \subseteq R, S \text{ finite}\}$
 the set of all limited distinction functions over $\mathcal{R}$.
\end{define}
To equip the new model with the dubious skill of distinction, we introduce oracle machines:
\begin{define}[Limited Distinction Oracle]
An SRTM $M$ over $\mathcal{R}$ \emph{with oracle access to }$\limrec$ is an 8-tuple $M=(R', Q, \Sigma_{I}, \Sigma, \iota, \mathcal{X}, \sqcup, \delta \cup \delta')$, where 
$(R', Q, \Sigma_{I}, \Sigma, \iota, \mathcal{X}, \sqcup, \delta)$ is an SRTM,
 $\delta' \subseteq \left(Q \times \Sigma \right)\times \left(Q \times \Sigma\right) \times \{-1,1\} \times (R' \cup \{\mathcal{X}\} \cup \limrec)$ is a labeled transition relation where the last 
component is the label, and  $|\delta \cup \delta'| <\infty$. 
Here, if the label of a transition is $\rec{S}$ and the semiring value at the current head position is $r$, then $M$ transitions with weight $\rec{S}(r)$.
Given that $\mathcal{C}$ is a semiring complexity class over semiring $\mathcal{R}$ defined in terms of SRTMs under a time bound or  a tape size bound, $\mathcal{C}^{\limrec}$ is the set of SRTMs over $\mathcal{R}$ \emph{with oracle access to $\limrec$} that satisfies the same resource bound.
\end{define}

This allows us to prove that limited distinction is the old models' only additional capability, in~\Cref{comparison}.
In turn, we can compensate the lack of functions such as $f$ from Example~\ref{ex:conditional-product} in \NPoldinf{$\mathcal{R}$} by considering its closure under Karp-reductions adapted to the semiring context:

\begin{define}[Karp Surrogate-Reduction~\cite{DBLP:journals/jair/EiterK23}]
For a finite alphabet $\Sigma$ such that $S,V,E\not\in \Sigma$, the surrogate alphabet $s_{S,V,E}(\Sigma)$ is given by $(\Sigma \cup \{SV^nE \mid n \in \mathbb{N}\})^{*}$. Here, $SV^nE$ is a surrogate for the semiring with index $n$, that is, for a substitution $\sigma\colon \mathbb{N} \longrightarrow R$ and $x \in s_{S,V,E}(\Sigma)$, we let $\sigma(x)$ be the string in $(\Sigma \cup R)^*$ such that every occurrence of $SV^nE$ in $x$ is replaced by $\sigma(n)$.

Let $f_1$ and $f_2$ be functions $f_i\colon (\Sigma \cup R)^{*} \longrightarrow R, i = 1,2$. 
Then a Karp \emph{s-reduction} from $f_1$ to $f_2$ is a polynomial time computable function $T\colon s_{S,V,E}(\Sigma) \longrightarrow s_{S,V,E}(\Sigma)$ with a finite set $R' = \{r_1,\ldots,r_k\}\subseteq R$, $k\geq 0$, of semiring values such that for all $x \in s_{S,V,E}(\Sigma)$ and $\sigma\colon \mathbb{N} \longrightarrow R$ such that $\sigma(i{-}1)=r_i$, $1\leq i \leq k$, it holds that $f_1(\sigma(x)) = f_2(\sigma(T(x)))$. 
\end{define}

Intuitively, a Karp s-reduction is a polynomial-time transformation of the input with semiring values that is agnostic about the specific semiring values on the tape. For an example, it is not hard to see that upon such a reduction, the function $f$ from Example~\ref{ex:conditional-product} can be computed by some SRTM according to \cite{DBLP:journals/jair/EiterK23}. We specifically  consider the \emph{closure}\/ of \NPoldinf{$\mathcal{R}$} under Karp s-reductions, which is defined in general as follows.
\begin{define}[Karp s-reduction Closure]
Let $\mathcal{C}$ be a semiring complexity class over the semiring $\mathcal{R}$ defined in terms of SRTMs. Then $\karp(\mathcal{C})$ consists of all functions $f$ that can be Karp s-reduced to some $g \in \mathcal{C}$.
\end{define}

The main result of this section is then that 
Karp s-reduction closure in the one model and providing distinction oracle access in the other yields the same function expressiveness.  The proof of the result is in~\cref{appendix:full_proofs}.

\begin{theorem}[\NPoldinf{$\mathcal{R}$} versus \NPnewinf{$\mathcal{R}$}]\label{comparison}
For any commutative semiring $\mathcal{R}$, 
$\karp(\NPoldinf{\mathcal{R}}) = \NPnewinf{\mathcal{R}}^{\limrec}$.
\end{theorem}
We are restricted to commutative semirings  since $\NPoldinf{\mathcal{R}}$ is only defined for these.

\section{Relation of \NPnewinf{$\mathcal{R}$} to Other Quantitative Complexity Classes}
%

Having clarified the relation between \NPoldinf{$\mathcal{R}$} and \NPnewinf{$\mathcal{R}$}, we can clearly see that \NPoldinf{$\mathcal{R}$} lacks some important properties. On the one hand, it is not closed under Karp s-reductions, on the other hand, it exhibits undesirable distinction capabilities. 
With this in mind, we consider the two claims of~\cite{DBLP:journals/jair/EiterK23} regarding the relationship between \NPoldinf{$\mathcal{R}$} and other quantitative complexity classes in the literature that are falsified by these shortcomings and properly establish them for \NPnewinf{$\mathcal{R}$}.

\noindent{\bf Traditional complexity classes.}~Eiter and Kiesel presented in \cite{DBLP:journals/jair/EiterK23} a host of results that linked traditional complexity classes such as \#P, GapP, OptP, FP$^{\mathrm{NP}}_\|$, FPSPACE(poly) etc.\ to \NPoldinf{$\mathcal{R}$} for suitable semirings $\mathcal{R}$.  
The derivation of the results hinged on the completeness of the satisfiability problem over $\mathcal{R}$ (i.e., whether a non-zero value is computed) under Karp s-reductions, as the prototypical \NPoldinf{$\mathcal{R}$}-complete problem. The proof used the following 
lemma, here stated for \NPnewinf{$\mathcal{R}$}.
%
\begin{lemma}[claimed as Lemma~17 in \cite{DBLP:journals/jair/EiterK23} for \NPoldinf{$\mathcal{R}$}]
\label{lem:vred}
Let $\mathcal{R}$ be a semiring and $f_i \colon (\Sigma \cup R)^{*} \longrightarrow R, i =1,2$. If $f_2 \in \NPnewinf{\mathcal{R}}$ and $f_1$ is Karp s-reducible to $f_2$ then $f_1 \in \NPnewinf{\mathcal{R}}$.
\end{lemma}
That is, the claim was that \NPoldinf{$\mathcal{R}$} is closed under Karp s-reductions; as we have seen, this is not the case. The reason is that in the SRTM model of \cite{DBLP:journals/jair/EiterK23}, we do not have the possibility to navigate to the $(n-k)^{\text{th}}$ semiring value on the tape, as required in the proof. However, as we prove in~\Cref{appendix:fixed_proofs}, the above statement for $\NPnewinf{\mathcal{R}}$ is correct.

 \emph{Weighted  quantified Boolean logic} over a  commutative  semiring  $\mathcal{R}$ 
 uses a set $\mathcal{V}$  of propositional variables and obeys the following grammar for \emph{weighted Quantified Boolean Formulas}: $\alpha ::= k \mid v \mid \neg v \mid \alpha \fplus \alpha \mid \alpha \fmal \alpha \mid \bigfplus v \alpha \mid \bigfmal v \alpha$,
where $k \in R$ and $v \in \mathcal{V}$. Given a weighted QBF $\alpha$ over a commutative semiring $\mathcal{R} = (R, \fplus, \fmal, \szero, \sone)$ and variables from $\mathcal{V}$ as well as a propositional interpretation $\mathcal{I}$ of $\mathcal{V}$, the semantics $\llbracket \alpha \rrbracket_{\mathcal{R}}(\mathcal{I})$ of $\alpha$ over $\mathcal{R}$ w.r.t.\ $\mathcal{I}$ is 
as follows:
 \vspace{-0.5cm}
\begin{align*}
    \llbracket k \rrbracket_{\mathcal{R}} (\mathcal{I})\ \  &= \ \ k\\
    \llbracket l \rrbracket_{\mathcal{R}} (\mathcal{I})\ \  &= \ \ \left\{\begin{array}{cc}
        \sone & l \in \mathcal{I} \\
        \szero & \text{ otherwise. }
    \end{array} \right. (l \in \{v, \neg v\})\\
    \llbracket \alpha_1 \fplus \alpha_2 \rrbracket_{\mathcal{R}} (\mathcal{I}) &= \llbracket \alpha_1\rrbracket_{\mathcal{R}} (\mathcal{I}) \fplus \llbracket\alpha_2 \rrbracket_{\mathcal{R}} (\mathcal{I}) &  \hspace{-1cm} \llbracket \alpha_1 \fmal \alpha_2 \rrbracket_{\mathcal{R}} (\mathcal{I}) &= \llbracket \alpha_1\rrbracket_{\mathcal{R}} (\mathcal{I}) \fmal \llbracket\alpha_2 \rrbracket_{\mathcal{R}} (\mathcal{I})
    \\
    \llbracket \bigfplus v \alpha \rrbracket_{\mathcal{R}} (\mathcal{I}) &= \llbracket \alpha \rrbracket_{\mathcal{R}} (\mathcal{I}_{v}) \fplus \llbracket \alpha \rrbracket_{\mathcal{R}} (\mathcal{I}_{\neg v})
      & \hspace{-1cm} \llbracket \bigfmal v \alpha \rrbracket_{\mathcal{R}} (\mathcal{I}) &= \llbracket \alpha \rrbracket_{\mathcal{R}} (\mathcal{I}_{v}) \fmal \llbracket \alpha \rrbracket_{\mathcal{R}} (\mathcal{I}_{\neg v})
\nop{******** non-compact form of sums/products
     \llbracket \alpha_1 \fplus \alpha_2 \rrbracket_{\mathcal{R}} (\mathcal{I})\ \  &= \ \ \llbracket \alpha_1\rrbracket_{\mathcal{R}} (\mathcal{I}) \fplus \llbracket\alpha_2 \rrbracket_{\mathcal{R}} (\mathcal{I}) \\
    \llbracket \alpha_1 \fmal \alpha_2 \rrbracket_{\mathcal{R}} (\mathcal{I})\ \ &= \ \ \llbracket \alpha_1\rrbracket_{\mathcal{R}} (\mathcal{I}) \fmal \llbracket\alpha_2 \rrbracket_{\mathcal{R}} (\mathcal{I})
    \\ 
    \llbracket \bigfplus v \alpha \rrbracket_{\mathcal{R}} (\mathcal{I})\ \ &= \ \ \llbracket \alpha \rrbracket_{\mathcal{R}} (\mathcal{I}_{v}) \fplus \llbracket \alpha \rrbracket_{\mathcal{R}} (\mathcal{I}_{\neg v})
    \\
    \llbracket \bigfmal v \alpha \rrbracket_{\mathcal{R}} (\mathcal{I})\ \ &= \ \ \llbracket \alpha \rrbracket_{\mathcal{R}} (\mathcal{I}_{v}) \fmal \llbracket \alpha \rrbracket_{\mathcal{R}} (\mathcal{I}_{\neg v})
********} 
\end{align*}
where $\mathcal{I}_{v}= \mathcal{I}\setminus\{\neg v\}\cup \{v\}$ and $\mathcal{I}_{\neg v}= \mathcal{I}\setminus\{v\}\cup \{\neg v\}$.


We further focus on $\bigfplus$BFs, i.e., the weighted fully quantified BFs that contain only sum quantifiers (i.e.\ $\bigfplus v$) and we introduce their evaluation problem as 
\textsc{SAT($\mathcal{R}$)}: given a $\bigfplus$BF $\alpha$ over $\mathcal{R}$ compute $\llbracket \alpha \rrbracket_{\mathcal{R}}(\emptyset)$.
The \NPoldinf{$\mathcal{R}$}-completeness of this problem was claimed in Theorem~18 of~\cite{DBLP:journals/jair/EiterK23} and is here corrected by considering \NPnewinf{$\mathcal{R}$}.
\begin{theorem}[claimed as Theorem~18 in \cite{DBLP:journals/jair/EiterK23} for \NPoldinf{$\mathcal{R}$}]
\label{thm:satnp}
\textsc{SAT}($\mathcal{R}$) is \NPnewinf{$\mathcal{R}$}-complete with respect to  Karp s-reductions, for every semiring $\mathcal{R}$.
\end{theorem}
While the proof for this theorem may look sound in~\cite{DBLP:journals/jair/EiterK23}, it implicitly assumes that semiring values are distinguishable during Karp s-reductions, which they are not. 

Inspection reveals that merely Lemma~17 and Theorem~18 of~\cite{DBLP:journals/jair/EiterK23} are flawed. 
In particular,
the results in Section~6 of~\cite{DBLP:journals/jair/EiterK23} about the relationship of \NPoldinf{$\mathcal{R}$} to classical counting and function computation classes for various semirings $\mathcal{R}$ are \emph{not} affected by these oversights. Furthermore, since the respective proofs there only refer to Lemma~17 and Theorem~18 but not to the definition of the SRTM model itself, they still work for any similar machine model that would satisfy  Lemma~17 and Theorem~18, and in particular for the STRM model that we introduce in this work; respective proofs for \cref{lem:vred,thm:satnp}
are provided in \cref{appendix:fixed_proofs}.
Consequently, 
all the results about linking \NPoldinf{$\mathcal{R}$} to classical complexity classes $\mathcal{C}$ in Section~6 of \cite{DBLP:journals/jair/EiterK23}
actually hold for \NPnewinf{$\mathcal{R}$}, spanning a rich landscape for different semirings $\mathcal{R}$.

\noindent
{\bf \texorpdfstring{\NPnewinf{$\mathcal{R}$}}{NP(R)\_{new}}-complete problems over input alphabet \texorpdfstring{$\emptyset$}{\{\} }.}
If we consider functions $f\colon R^* \longrightarrow R$ that are computable in \NPnewinf{$\mathcal{R}$}, it may seem that we cannot add much qualitative information in  the input string as SRTMs can only distinguish semiring values on the tape in a very limited manner. That is, we can transition with the value on the tape; if it happens to be zero, then the remaining computation does not have an effect. So we can merely ``check'' whether a semiring value on the tape is non-zero.

However, this is sufficient to gain \NPnewinf{$\mathcal{R}$}-completeness
(for the proof, see \cref{appendix:full_proofs}).

\vspace{0pt}

\begin{proposition}
For every semiring $\mathcal{R}$, there exists a function $f\colon R^* \longrightarrow R$ that is \NPnewinf{$\mathcal{R}$}-complete with respect to Karp s-reductions.
\end{proposition}

\noindent{\bf Class \NPfinite{$\mathcal{R}$}.}~
Let us denote by \NPfinite{$\mathcal{R}$} the complexity class introduced as NP$(\mathcal{R})$ for a semiring $\mathcal{R}$ in~\cite{DBLP:conf/mfcs/BadiaDNP24}, i.e., the class of all functions without semiring values in the input, i.e. $f\colon\Sigma_{I}^* \longrightarrow R$, where $\Sigma_{I}$ is the input alphabet as in \cref{srtm}, computable by some SRTM in polynomial time over $\mathcal{R}$.
Clearly, \NPfinite{$\mathcal{R}$} is subsumed by \NPnewinf{$\mathcal{R}$}; as it turns out, it is strictly less expressive. The proof of the 
following proposition is relegated to~\cref{appendix:full_proofs}.
\begin{proposition}
There exists a function in $\NPnewinf{\mathcal{R}}$ that is not in $\NPfinite{\mathcal{R}}$.
\end{proposition}

\section{Conclusion}

Our Definition~\ref{srtm} for Semiring Turing Machines improves the limitations of  the central definition from~\cite{DBLP:journals/jair/EiterK23} and allows us to recover all the complexity results established there, linking computations over semirings to many classical complexity classes. \cref{comparison} provides the exact relation between \NPoldinf{$\mathcal{R}$}  (from~\cite{DBLP:journals/jair/EiterK23}) and  our new class \NPnewinf{$\mathcal{R}$}. Furthermore, Theorem~\ref{fagin} establishes a logical characterization of the  quantitative complexity class 
\NPnewinf{$\mathcal{R}$} using a  weighted existential second-order logic. This solves an open problem posed in~\cite{DBLP:journals/jair/EiterK23}. 

The new complexity class \NPnewinf{$\mathcal{R}$} may be fruitfully used for analyzing problems in applications. 
As an example, we briefly illustrate this for establishing the combined complexity of conjunctive queries (CQs) and unions of conjunctive queries (UCQs) in databases over semirings (for details see Appendix~\ref{appendix:UCQ}). 
For CQs, this problem can be formulated as follows:
given a positive, single-rule datalog program $\Pi = \{ q \leftarrow r_1(\vec{Y_1}), \dots, r_n(\vec{Y_n}) \}$
and a semiring-weighted extensional database $D$ over a commutative semiring $\mathcal{R}$ as input, compute the label of $q$ (the semantics of the program --or CQ-- is given as in~\cite{Green:2011} in terms of sums of products). 
Similarly, for UCQs the problems can be formulated as follows:
given a set of positive, single-rule datalog programs $\{\Pi_1, \dots, \Pi_m\}$
and a semiring-weighted extensional database $D$ over a commutative semiring $\mathcal{R}$ as input, compute the sum of the labels of $q_1, \dots, q_n$ (the semantics of this set of programs --or UCQ-- is given as in~\cite{Green:2011} in terms of sums of products). 

Putting together a result from~\cite{DBLP:journals/jair/EiterK23} on CQs and 
new results on UCQs (see Appendix~\ref{appendix:UCQ}), 
we obtain the following result.

\begin{theorem}
\label{thm:harddata}
For any commutative semiring $\mathcal{R}$, the combined complexity of both CQs and UCQs is \NPnewinf{$\mathcal{R}$}-complete with respect to Karp s-reductions.
\end{theorem}

Consequently, we obtain that the combined complexity of CQs and UCQs in bag semantics (with semiring is $\mathbb{N}$) is $\#$P-complete and that for the provenance semiring of polynomials with coefficients on the natural numbers the problems are $\#$P-hard (cf.\ also~\cite[footnote 8]{Green:2011}).
\vspace{0pt}

\noindent{\bf Outlook.} Lines of future  research include the study of further quantitative complexity classes defined on SRTMs, such as a natural generalization of the classical counting class $\#$L,  and corresponding Fagin-style characterizations (such a characterization for $\#$L itself  was obtained in \cite{arenas19descriptive}).


\bibliography{fagin}

\appendix

\section{Appendix}

\markboth{Appendix}{Appendix}

\subsection{Comparison to \texorpdfstring{$K$}{K}-Turing Machines} \label{appendix:comparision}
Recently, a model of computation taking  finite sequences of  elements of a given fixed semiring has been introduced in~\cite{DBLP:journals/corr/abs-2502-12939} under the name \emph{$K$-Turing Machine} to deal with the data complexity problem for first-order logic with semiring semantics~\cite{G-T}.  While both our model of computation and the $K$-TMs machines of~\cite{DBLP:journals/corr/abs-2502-12939} allow for computations with semiring values, there are multiple differences between the two.
Although the $K$-TMs model is deterministic in~\cite{DBLP:journals/corr/abs-2502-12939}, a non-deterministic version appears in  ~\cite{barlag2025logicalapproachesnondeterministicpolynomial}. However,  the complexity classes introduced for $K$-TMs machines are formed by decision problems, not function problems as in our case. As such, connections to  \#P and other counting classes are not within the purview of their machine model. By contrast, we defined Semiring Turing Machines in such a manner that they generalize NP to the semiring setting, in such a manner that \#P corresponds to the problems solvable in polynomial time by SRTM's over the natural numbers.  This difference explains as well why the existential second-order logic   in \cite{barlag2025logicalapproachesnondeterministicpolynomial} involves only quantification over Boolean relations.

While $K$-TMs machines seem to be less powerful than SRTMs in this aspect, there are others with respect to which they are definitely much stronger. The main point here is that $K$-TMs machines allow for the reuse of the results of semiring computations by storing them on the tape and in \emph{registers}. This makes it possible to compute $a_n$ in polynomial time in $n$, where 
$a_0 = 2, a_{n + 1} = a_{n}\cdot a_{n}$.
Clearly $a_n = 2^{2^{n}}$; hence, $K$-TMs machines over the natural numbers can compute a value that is double exponential in the input, which is not possible using SRTMs.

Additionally, $K$-TMs allow for branching based on the comparisons of semiring values in the registers and on the tape. While SRTMs can simulate such a behavior to a certain extent, semiring values that are on the tape cannot be compared to other values. The only thing we can do in this respect is transition with the weight on the tape to ensure that it is not zero.\

\subsection{Definitions from~\cite{DBLP:journals/jair/EiterK23}}
\label{appendix:old_definitions}
All of the contents of this subsection are copied from~\cite{DBLP:journals/jair/EiterK23}.
First, the old definition of Semiring Turing Machines:
\begin{define}[SRTM]
A \emph{Semiring Turing Machine} is a 7-tuple $M=(\mathcal{R}, R', Q, \Sigma, \iota, \sqcup, \delta)$, where
\begin{itemize}
    \item $\mathcal{R}$ is a commutative semiring,
    \item $R' \subseteq R$ is a finite set of semiring values, (intuitively, this is a set of fixed values that are ``known'' to $M$)
    \item $Q$ is a finite set of states,
    \item $\Sigma$ is a finite set of symbols (the tape alphabet),
    \item $\iota \in Q$ is the initial state,
    \item $\sqcup \in \Sigma$ is the blank symbol,
    \item $\delta \subseteq \left(Q \times (\Sigma \cup R) \right)\times \left(Q \times (\Sigma \cup R)\right) \times \{-1,1\} \times R$ is a weighted transition relation, where the last entry of the tuple is the weight. For each $((q_1, \sigma_1), (q_2, \sigma_2), d , r) \in \delta$ the following holds:
    \begin{enumerate}
        \item \emph{$M$ cannot write or overwrite semiring values}: \\
        if $\sigma_1 \in R$ or $\sigma_2 \in R$, then $\sigma_1 = \sigma_2$,
        \item \emph{$M$ can only make a transition with weight $r$ when $r$ is from $R'$ or under the head:} \\
        $r \in R'$ or $r = \sigma_1 \in R$
        \item \emph{$M$ cannot discriminate semiring values:} \\
        if $\sigma_1 \in R$, then
        \begin{enumerate}
            \item for all $\sigma_1' \in R$ 
        we have $((q_1, \sigma_1'),$ $(q_2, \sigma_1'), d, \sigma_1') \in \delta$ or
            \item for all $\sigma_1' \in R$ 
        we have $((q_1, \sigma_1'),$ $(q_2, \sigma_1'), d, r) \in \delta$.
        \end{enumerate}
    \end{enumerate}
\end{itemize} 
\end{define}

Next, we need the definition of the prototypical complete problem for our complexity class, for the last section of the appendix:
\begin{define}[Syntax]
Let $\mathcal{V}$ be a set of propositional variables and $\mathcal{R} = (R, \fplus, \fmal, \szero, \sone)$ be a commutative semiring. A \emph{weighted} QBF over $\mathcal{R}$ is of the form $\alpha$ given by the grammar
\begin{align*}
    \alpha ::= k &\mid v \mid \neg v \mid \alpha \fplus \alpha \mid \alpha \fmal \alpha \mid \bigfplus v \alpha \mid \bigfmal v \alpha
\end{align*}
where $k \in R$ and $v \in \mathcal{V}$. A variable $v \in \mathcal{V}$ is \emph{free} in a weighted QBF $\alpha$, if $\alpha$ has $v$ not in the scope of a quantifier $\bigfplus v$ or $\bigfmal v$.
A weighted \emph{fully quantified Boolean Formula} is a weighted QBF without free variables.
\end{define}
\begin{define}[Semantics]
\label{def:wQBF}
Given a weighted QBF $\alpha$ over a commutative semiring $\mathcal{R} = (R, \fplus, \fmal, \szero, \sone)$ and variables from $\mathcal{V}$ as well as a propositional interpretation $\mathcal{I}$ of $\mathcal{V}$, the semantics $\llbracket \alpha \rrbracket_{\mathcal{R}}(\mathcal{I})$ of $\alpha$ over $\mathcal{R}$ w.r.t.\ $\mathcal{I}$ is defined as follows:
\begin{align*}
    \llbracket k \rrbracket_{\mathcal{R}} (\mathcal{I}) &= k\\
    \llbracket l \rrbracket_{\mathcal{R}} (\mathcal{I}) &= \left\{\begin{array}{cc}
        \sone & l \in \mathcal{I} \\
        \szero & \text{ otherwise. }
    \end{array} \right. (l \in \{v, \neg v\})\\
    \llbracket \alpha_1 \fplus \alpha_2 \rrbracket_{\mathcal{R}} (\mathcal{I}) &= \llbracket \alpha_1\rrbracket_{\mathcal{R}} (\mathcal{I}) \fplus \llbracket\alpha_2 \rrbracket_{\mathcal{R}} (\mathcal{I})\\
    \llbracket \alpha_1 \fmal \alpha_2 \rrbracket_{\mathcal{R}} (\mathcal{I}) &= \llbracket \alpha_1\rrbracket_{\mathcal{R}} (\mathcal{I}) \fmal \llbracket\alpha_2 \rrbracket_{\mathcal{R}} (\mathcal{I})\\
    \llbracket \bigfplus v \alpha \rrbracket_{\mathcal{R}} (\mathcal{I}) &= \llbracket \alpha \rrbracket_{\mathcal{R}} (\mathcal{I}_{v}) \fplus \llbracket \alpha \rrbracket_{\mathcal{R}} (\mathcal{I}_{\neg v})\\
    \llbracket \bigfmal v \alpha \rrbracket_{\mathcal{R}} (\mathcal{I}) &= \llbracket \alpha \rrbracket_{\mathcal{R}} (\mathcal{I}_{v}) \fmal \llbracket \alpha \rrbracket_{\mathcal{R}} (\mathcal{I}_{\neg v})
\end{align*}
where $\mathcal{I}_{v}= \mathcal{I}\setminus\{\neg v\}\cup \{v\}$ and $\mathcal{I}_{\neg v}= \mathcal{I}\setminus\{v\}\cup \{\neg v\}$.
\end{define}
Weighted QBFs generalize QBFs in negation normal form (NNF), as negation is only allowed in front of 
variables. 
Intuitively, allowing negation in front of complex formulas would add the ability to test whether the value of a weighted QBF is zero. We can only test whether an atomic formula is false. Therefore, our variant is or at least seems less expressive. However, it fits better into the context of the problems we consider, where it also is not possible to perform such ``zero-tests'' for complex expressions.

Here, we further focus on $\bigfplus$BFs, i.e., the weighted fully quantified BFs that contain only sum quantifiers (i.e.\ $\bigfplus v$) and we introduce their evaluation problem as 
\begin{center}
\textsc{SAT($\mathcal{R}$)}: given a $\bigfplus$BF $\alpha$ over $\mathcal{R}$ compute $\llbracket \alpha \rrbracket_{\mathcal{R}}(\emptyset)$.
\end{center}

\subsection{Full versions of the proofs}
\label{appendix:full_proofs}

We defined our logic in two stages, a Boolean logic and a weighted part. We just note here that we could have defined the logic $\mathrm{wESO}[\mathcal{R}]$ without using Boolean values. For this, we would employ in the "Boolean part" only conjunction, universal quantifications and negations, since these suffice to express all formulas of second order logic equivalently and products and negations of $0$ and $1$ in the semiring yield  $0$  and $1$, respectively $1$  and  $0$. This would provide a proper generalization e.g. of Fagin's theorem into our weighted version. We chose the present way because it is very natural. Also, practitioners can use the Boolean part as they want to and are used to and the weighted part in addition when needed.  

\setcounter{theorem}{11}

\begin{theorem}
For each semiring $\mathcal{R}$, the logic $\mathrm{wESO}[\mathcal{R}]$ captures the quantitative complexity class \NPnewinf{$\mathcal{R}$}  for ordered structures, and for all finite structures if  $\mathcal{R}$  is commutative and idempotent. 
\end{theorem}

\newcommand{\mac}{M}
\begin{proof}
In order to establish (1) from Definition~\ref{ch},
we construct, for every $\mathrm{wESO}$-formula $\phi$, a polynomial-time SRTM $\mac$ with $\|\phi\|(\mfa) = \|\mac\|(\text{enc}(\mfa))$ where $\text{enc}(\mfa)$ is as follows. If $v_1, \ldots, v_m$ are the free variables of $\phi$, we encode every input structure $\mathfrak{A}= (A, R_1^\mathfrak{A}, \dotsc, R_{\ell_1}^\mathfrak{A},W_1^\mathfrak{A}, \dotsc, W_{\ell_2}^\mathfrak{A})$ and free variable assignments $v_1 = a_1, \dotsc, v_m = a_m$ for $\mac$ by
  \[\text{enc}(\mathfrak{A})= 0^n 1
                              \text{enc}(R_1^\mathfrak{A}) \cdot \dots \cdot \text{enc}(R_{\ell_1}^\mathfrak{A})
                              \text{enc}(W_1^\mathfrak{A}) \cdot \dots \cdot \text{enc}(W_{\ell_2}^\mathfrak{A})
                              \text{enc}(a_1) \cdot \dots \cdot \text{enc}(a_m), \]
  where $n = |A|$ and $\text{enc}(a_i) = \text{enc}(\{a_i\})$ if $v_i$ is a first-order variable.
  Let us assume that $A = \{0, \dotsc, n-1\}$.
  We proceed by induction on the structure of formulas
  and begin by showing that for every Boolean formula $\beta$,
  there exists a deterministic polynomial time Turing machine $\mac$
  such that $(\mathfrak{A}, a_1, \ldots, a_m) \models \beta$ iff $\mac$
  accepts $(\mathfrak{A}, a_1, \ldots, a_m)$,
  see also~\cite[Proposition 6.6]{DBLP:books/sp/Libkin04} (see the full argument in~\cref{appendix:full_proofs}).
  All of these constructions produce deterministic polynomial time Turing machines.
  Let $\mac_\beta$ be the Turing machine constructed for $\beta$.
  We obtain an SRTM $\mac$ from $\mac_\beta$
  with $\|\beta\|(\mfa) = \|\mac\|(\text{enc}(\mfa))$ by defining the weight of every transition
  by $\sone$.
  Note that as $\mac_\beta$ is deterministic,
  there is exactly one run with weight $\sone$ in $\mac$
  for every input accepted by $\mac_\beta$.
  
  We continue with the weighted formulas $\phi$.
  The case $\phi = \beta$ is already covered.
  For $\phi = \sre$, we construct an SRTM which
  accepts in a single transition with weight $\sre$ to a final state, i.e., transitions to a state without outgoing transitions.
  \begin{itemize}
    \item For \(\phi = W(x_1, \dotsc, x_l)\),
          we construct an SRTM which moves the head to the position
          where \(W(x_1, \dotsc, x_m)\) is stored on the input tape
          using transitions of weight \(\sone\)
          and then applies a transition
          \((q,\mathcal{X},q_f,\mathcal{X},1,\mathcal{X})\),
          where \(q\) is the state reached after moving to the correct tape cell, into a state \(q_f\) without outgoing transitions.
    \item For $\phi = \psi \fplus \zeta$,
          we let $\mac_\psi$ and $\mac_\zeta$ be
          the SRTMs for $\psi$ and $\zeta$, respectively.
          We construct an SRTM $\mac$ for $\phi$ which
          non-deterministically simulates either $\mac_\psi$ or $\mac_\zeta$
          on the input.
          The simulations are started by a single transition of weight $\sone$ from the new initial state
          into the initial state of either $\mac_\psi$ or $\mac_\zeta$.
          Thus, every run of either $\mac_\psi$ or $\mac_\zeta$
          on the input is simulated by exactly one run of $\mac$.
    \item For $\phi = \psi \fmal \zeta$,
          we let $\mac_\psi$ and $\mac_\zeta$ be
          the SRTMs for $\psi$ and $\zeta$, respectively.
          We construct an SRTM $\mac$ for $\phi$ which
          first simulates $\mac_\psi$ and then $\mac_\zeta$ on the input.
          All transitions outside of the simulations,
          e.g., preparing the input for a simulation and clearing the tape for the second simulation,
          have weight $\sone$ and are deterministic.
          Thus, every combination of a run $r_\psi$ of $\mac_\psi$
          and a run $r_\zeta$ of $\mac_\zeta$ on the input
          is simulated by exactly one run of $\mac$
          and the weight of this run is the product of the
          weights of $r_\psi$ and $r_\zeta$.
          As multiplication distributes over addition,
          $\mac$ recognizes $\phi$.
    \item If $\phi = \bigfplus x.\psi$,
          we let $\mac_\psi$ be
          the SRTM for $\psi$.
          We construct an SRTM $\mac$ for $\phi$ which non-deterministically
          guesses an element $a \in A$ and simulates
          $\mac_\psi$ on the input $\tuple{\mathfrak{A}, a, a_1, \dotsc, a_l}$.
          We ensure that for every $a \in A$, there is exactly one run of $\mac$
          which guesses $a$ and that all transitions which prepare the simulation
          have weight $\sone$.
          Thus, for every choice of $a \in A$,
          every run of $\mac_\psi$ on $\tuple{\mathfrak{A}, a, a_1, \dotsc, a_l}$
          is simulated by exactly one run of $\mac$.
    \item If $\phi = \bigfmal x.\psi$,
          we let $\mac_\psi$ be
          the SRTM for $\psi$.
          We construct an SRTM for $\phi$ which iterates
          over all elements $a \in A$ and simulates
          $\mac_\psi$ on the input $\tuple{\mathfrak{A}, a, a_1, \dotsc, a_l}$.
          All transitions outside of the simulations,
          e.g., preparing the input for a simulation and clearing the tape for the next simulation,
          have weight $\sone$ and are deterministic.
          Thus, every combination of runs $r_0, \dotsc, r_{n-1}$ of $\mac_\psi$
          on $\tuple{\mathfrak{A}, 0, a_1, \dotsc, a_l}$, $\dotsc$,  $\tuple{\mathfrak{A}, n-1, a_1, \dotsc, a_l}$, respectively,
          is simulated by exactly one run of $\mac$
          and the weight of this run is the product of the
          weights of $r_0, \dotsc, r_{n-1}$.
          As multiplication distributes over addition,
          $\mac$ recognizes $\phi$.
    \item If $\phi = \bigoplus X.\psi$,
          we proceed like in the case of $\bigoplus x.\psi$
          but guess a relation $R$ of appropriate arity for $X$
          instead of an element $a \in A$.
  \end{itemize}

Now, for (2) from Definition~\ref{ch}, let
$M=(R', Q, \Sigma_{I}, \Sigma, \iota, \mathcal{X}, \sqcup, \delta)$
be an SRTM whose runs are bounded by time $n^k$ where $n$ is the length of the input (the encoding of a structure from  $\str{\tau}_<$ and $k$ is bigger than the maximum of the arities of the relational symbols in $\tau$). We will construct a $\mathrm{wESO}$-formula $\phi$ such that for ordered structures we have  $\|\phi\|= \|\mac\|$. Let $Q=\{q_0, \dots, q_{m-1}\}$.
We assume the same encoding for $\mathfrak{A}\in \str{\tau}_<$ as before
such that the encoding consists of symbols \(0\) and \(1\) and elements from \(\Sr\).
Thus, we may choose the tape alphabet $\Sigma$
in our SRTM to be $\{0,1,\mathcal{X}, \sqcup\}$.

The first task is to build a Boolean second-order formula $$\psi(T_0, T_1, T_2, T_3, H_{q_0}, \dots, H_{q_{m-1}}),$$  without second-order quantifiers but where $T_0, T_1, T_2, T_3, H_{q_0}, \dots, H_{q_{m-1}}$ are new second-order variables, such that there is a one-to-one correspondence between the computation paths for the input $\text{enc}(\mathfrak{A})$ and the expansions of the model $\mathfrak{A}$ that satisfy  $\psi(T_0, T_1, T_2, T_3, H_{q_0}, \dots, H_{q_{m-1}})$.

Recall that $\mathfrak{A}$ with $|A|=n$ is linearly ordered by $<$. We will represent the $n^k$ time steps and tape cells as the elements of the set $A^k$, so $k$-tuples from $A$. From $<$ we can define in first-order logic an associated successor relation $\mathtt{Succ}$, as well as the bottom $\bot$ and top element $\top$. With this at hand, if $\overline{x}, \overline{y}$ are $k$-tuples of variables we can define a successor relation on the elements of $A^k$ by 
$$\textstyle \overline{x} = \overline{y} +1 := \bigvee_{i\leq k} (\bigwedge_{j<i} (x_j=\bot \wedge y_j=\top) \wedge \mathtt{Succ}(x_i,y_i) \wedge \bigwedge_{j>i} x_j=y_j).$$

Next, let us spell out the meaning of the predicates $T_0, T_1, T_2, T_3, H_{q_0}, \dots, H_{q_{m-1}}$:

\begin{itemize}
\item[(1)] $T_i (\overline{p}, \overline{t})$ $(i=0,1)$, where  $\overline{p}, \overline{t}$ are $k$-tuples of first-order variables, is meant to represent that at time $\overline{t}$ the position $\overline{p}$ of the tape contains the symbol $i$. $T_2$ does the same but for \(\mathcal{X}\) and \(T_3\) for the blank symbol.

\item[(2)] $H_{q_i} (\overline{p}, \overline{t})$ $(0\leq i \leq m-1)$   represents that at time $\overline{t}$ the machine is in state $q_i$ with its head in position $\overline{p}$.
\end{itemize}

The idea is that we will use the predicates  $T_i$ and $H_q$ to describe an accepting computation of $\mac$ started with input $\text{enc}(\mathfrak{A})$.
For simplicity, we would like to assume that all computations on a structure of size $n$ have length $n^k$,
i.e., that there are no shorter computations.
In order to do this,
we modify the SRTM by adding transitions which ``do nothing'' from all accepting states, i.e.,
all states without outgoing transitions.
The resulting SRTM has the same behavior as our given SRTM when only taking into account
computations of exactly length $n^k$, which is what the formula we construct will do.

Given $k$-tuples of variables $\overline{x}=x_1, \dots, x_k$ and $\overline{y}=y_1, \dots, y_k$, we write $\overline{x}\neq \overline{y}$ as an abbreviation for $\bigvee_{1\leq i\leq k}x_i\neq y_i$.
We let $\psi(T_0, T_1, T_2, T_3, H_{q_0}, \dots, H_{q_{m-1}})$ be the conjunction of the following:
\begin{itemize}
\itemsep=4pt
\item 
\quad $\forall\overline{p}\forall\overline{t} \bigvee_{i=0}^3 \Big( T_i(\overline{p}, \overline{t}) \land \lnot \bigvee_{j \neq i} T_j(\overline{p}, \overline{t}) \Big)$

\smallskip

``In every configuration each cell of the tape contains exactly one symbol from the alphabet $\Sigma$.'' 

\item \quad  $\forall\overline{t}\exists ! \overline{p} (\bigvee_{q \in Q} H_q (\overline{p}, \overline{t})) \wedge \forall \overline{p}\forall\overline{t} (\bigwedge_{q\neq q'\in Q} (\neg H_q(\overline{p}, \overline{t}) \vee \neg H_{q'}(\overline{p}, \overline{t}))) $
\smallskip

``At any time the head is at exactly one position and the machine $\mac$ is in exactly one state.'' 




\item\quad $\forall \overline{t} \bigvee_{\substack{(p,a,q,b,D,v) \in \delta}} \theta_{(p,a,q,b,D,v)}(\overline{t})
\lor \overline{t} = (\top, \dotsc, \top)$, where\\
 $\mbox{\quad} \theta_{(p,a,q,b,1,v)}(\overline{t}) :=
\exists \overline{p} \exists \overline{q} \exists \overline{s}
\Big( (\overline{p} = \overline{q} + 1) \land (\overline{s} = \overline{t} + 1 ) \land H_p(\overline{p},\overline{t}) \land T_a(\overline{p},\overline{t}) \land H_q(\overline{q}, \overline{s}) \\
 \phantom{\mbox{\quad} \theta_{(p,a,q,b,1,v)}(\overline{t}) :=
\exists \overline{p} \exists \overline{q} \exists \overline{s}
\Big( } \land T_b(\overline{p}, \overline{s})
\land \forall \overline{p}' (\overline{p}' \neq \overline{p} \to \bigwedge_{i=0}^3 (T_i(\overline{p}', \overline{t}) \leftrightarrow T_i(\overline{p}', \overline{s})))\, \Big)$,\\
$\mbox{\quad} \theta_{(p,a,q,b,-1,v)}(\overline{t}) :=
\exists \overline{p} \exists \overline{q} \exists \overline{s}
\Big( (\overline{p} = \overline{q} - 1) \land (\overline{s} = \overline{t} + 1 ) \land H_p(\overline{p},\overline{t}) \land T_a(\overline{p},\overline{t}) \land H_q(\overline{q}, \overline{s}) \\
\phantom{\mbox{\quad} \theta_{(p,a,q,b,-1,v)}(\overline{t}) :=
\exists \overline{p} \exists \overline{q} \exists \overline{s}
\Big( } \land T_b(\overline{p}, \overline{s})
\land \forall \overline{p}' (\overline{p}' \neq \overline{p} \to \bigwedge_{i=0}^3 (T_i(\overline{p}', \overline{t}) \leftrightarrow T_i(\overline{p}', \overline{s})))\Big).$

\smallskip 

``The configurations respect the transitions in $\delta$.'' Here, we abuse notation and identify \(2\) with \(\mathcal{X}\)
and \(3\) with \(\sqcup\)

\item 
\quad $H_{\iota}( \underbrace{\bot \cdot\dots \cdot \bot}_{k\text{-times}},  \underbrace{\bot \cdot\dots \cdot \bot}_{k\text{-times}}) \land
\forall \overline{p}.(\overline{p} < \bot^k + \top \to T_0(\overline{p})) \land
\forall \overline{p}.(\overline{p} = \bot^k + \top \to T_1(\overline{p})) \land\\
\mbox{\quad} \bigwedge_{i = 1}^{\ell_1} \forall x_1 \ldots \forall x_{r_i}. R_i(\overline{x}) \to T_1(\bot^k + \top + \sum_{j = 1}^{i-1} \top^{r_j} + \bar{x}, \bot^k) \land\\
\mbox{\quad} \bigwedge_{i = 1}^{\ell_1} \forall x_1 \ldots \forall x_{r_i}. \lnot R_i(\overline{x}) \to T_0(\bot^k + \top + \sum_{j = 1}^{i-1} \top^{r_j} + \bar{x}, \bot^k) \land\\
\mbox{\quad} \forall \overline{p} . (\bot^k + \top + \sum_{j = 1}^{\ell_1} \top^{r_j}) \leq \overline{p} \land \overline{p} < (\bot^k + \top + \sum_{j = 1}^{\ell_1} \top^{r_j} + \sum_{j = 1}^{\ell_2} \top^{w_j}) \to T_2(\overline{p}, \bot^k) \land\\
\mbox{\quad} \forall \overline{p} . (\bot^k + \top + \sum_{j = 1}^{\ell_1} \top^{r_j} + \sum_{j = 1}^{\ell_2} \top^{w_j}) \leq \overline{p} \to T_3(\overline{p}, \bot^k)$,

\smallskip

where $r_i$ is the arity of $R_i$ and $w_i$ the arity of $W_i$.

\smallskip

``At the initial time the tape contains $\text{enc}(\mathfrak{A})$ and it is in the initial state $\iota$.'' 

Note that we may define an addition of tuples using formulas like
\begin{align*}
  (\overline{q} = \overline{p} + \overline{x}) = \exists X \big(& \forall \overline{p}' ( \overline{p} \leq \overline{p}' \land \overline{p}' < \overline{q} \to \exists \overline{x}'. \overline{x}' \leq \overline{x} \land X(\overline{p}', \overline{x}') ) \land\\
  & \forall \overline{x}'. \overline{x}' \leq \overline{x} \to \exists \overline{p}' ( \overline{p} \leq \overline{p}' \land \overline{p}' < \overline{q} \land X(\overline{p}', \overline{x}') ) \land\\
 &  \forall \overline{p}' \forall \overline{p}'' \forall \overline{x}' \forall \overline{x}'' (X(\overline{p}', \overline{x}') \land X(\overline{p}'', \overline{x}'') \to (\overline{p}' = \overline{p}'' \leftrightarrow \overline{x}' = \overline{x}''))\big),
\end{align*}
where \(\overline{x}\) and \(\overline{x}'\) are tuples of length \(r\)
and \(X\) is of arity \(k+r\).
The formula above ensures the existence of a bijection (line 3) \(X\) from all positions between \(\overline{p}\) and \(\overline{q}\) (exclusive, line 1) onto all positions between \(\bot^r\) and \(\overline{x}\) (inclusive, line 2).

\end{itemize}

Finally, to obtain an $\mathrm{wESO}$-formula $\phi$ such that  $\| \phi\|(\mfa) = \|\mac\|(\text{enc}(\mfa))$, we consider the formula $\chi(\overline{t})$ which computes the weight of the transition applied at time step \(\overline{t}\) in the run encoded by $T_0, T_1, T_2, H_{q_0},\dots, H_{q_{m-1}}$:
\begin{align*}
  \chi(\overline{t}) &= (\overline{t} = \top^k) \fplus \chi_1(\overline{t}) \fplus \chi_2(\overline{t}),\\
  \chi_1(\overline{t}) &=  \bigfplus_{\substack{(p,a,q,b,D,\sre) \in \delta\\\sre\in\Sr}} \sre \fmal \theta_{(p,a,q,b,D,\sre)}(\overline{t}),\\
  \chi_2(\overline{t}) &= \bigfplus_{(p,a,q,b,D,\mathcal{X}) \in \delta} \chi_2'(\overline{t}) \fmal \theta_{(p,a,q,b,D,\mathcal{X})}(\overline{t}),\\
  \chi_2'(\overline{t}) &= \bigoplus \overline{p} . H(\overline{p}, \overline{t}) \otimes
    \bigoplus_{\substack{i = 1}}^{\ell_2} \bigoplus x_1 \dotsb \bigoplus x_{w_i} . W_i(\overline{x}) \otimes \beta_{W_i}(\overline{p}, \overline{x}),
\end{align*}
    where \(H(\overline{p}, \overline{t})\) is true iff the head is at position \(\overline{p}\) at time \(\overline{t}\), i.e., \(H(\overline{p}, \overline{t}) = \bigvee_{q \in Q} H_q(\overline{p}, \overline{t})\),
    and \(\beta_{W_i}(\overline{p}, \overline{x})\) is true iff
    the weight of \(W_i(\overline{x})\) is stored at position \(\overline{p}\) in the encoding \(\mathrm{enc}(\mathfrak{A})\),
    i.e., \(\beta_{W_i}(\overline{p}, \overline{x}) = (\overline{p} = \bot^k + \top + \sum_{j = 1}^{\ell_1} \top^{r_j} + \sum_{j = 1}^{i-1} \top^{w_j} + \overline{x})\).
Finally, we define $\phi$ as
\[\bigoplus T_0 \bigoplus T_1 \bigoplus T_2 \bigoplus T_3 \bigoplus H_{q_0} \ldots \bigoplus H_{q_{m-1}} \psi(T_0, T_1, T_2, T_3, H_{q_0}, \dots, H_{q_{m-1}})\fmal
\bigfmal \overline{t}. \chi(\overline{t}).\]

If  $\mathcal{R}$  is commutative and idempotent, we argue for arbitrary finite structures as described in Remark 13, using that the value  $\| \phi\|(\mfa) = \|\mac\|(\text{enc}(\mfa))$\ does not depend on the particular order of the structure $\mfa$.
\end{proof}

\setcounter{theorem}{26}
\begin{proposition}
\label{lemma:old_weaker}
For the semiring $\mathcal{R} = \mathbb{B}$, the function 
\begin{align*}
    f\colon& (\Sigma \cup R)^* \longrightarrow R, 
    x \mapsto \left\{ \begin{array}{cr}
        s_1 \fmal \dots \fmal s_n & \text{ if } x  = \sigma_1 \dotsb \sigma_n s_1 \dotsb s_n,\text{ where } \sigma_i \in \Sigma, s_i \in R, \\
        e_{\oplus} & \text{ otherwise.}
    \end{array}\right.
\end{align*}
cannot be computed by an SRTM according to the definitions of~\cite{DBLP:journals/jair/EiterK23}.
\end{proposition}
\begin{proof}[Proof (sketch)]
We restrict ourselves to inputs of the form $x  = \sigma_1 \dotsb \sigma_n s_1 \dotsb s_m$, where $s_i = 1$. So essentially, $f(x) = 1$ iff $n = m$ and we have a decision problem.  

We proceed by proof by contradiction. So we fix an SRTM $M$ that computes $f$. Note that since semiring values are observe-only (meaning we can only distinguish whether there is a semiring value or not) for the SRTM model  in \cite{DBLP:journals/jair/EiterK23}, and we cannot overwrite them, $M$ acts on the tape positions $n + 1$ to $n + m$ like a finite state automaton. Thus, we can apply the pumping lemma to derive that there is a $k$ such that if $n > k$ then there exists more than one $m$ such that $M$ accepts.

    
\end{proof}

\begin{proposition}
\label{proposition:old_stronger}
For the semiring $\mathcal{R} = \mathbb{N}$  and $r = 1$, the function 
\[
\text{dis}_r\colon (\Sigma \cup R)^* \longrightarrow R,\quad  x \mapsto \left\{ \begin{array}{cc}
        r & \text{ if } x_0 = r, \\
        x_0 \oplus r & \text{ if } x_0 \in R, x_0 \neq r\\
        e_{\oplus} & \text{ otherwise.}
    \end{array}\right.,
\]
cannot be computed by an SRTM according to the definitions of this paper.
\end{proposition}
\begin{proof}[Proof]

We proceed by proof by contradiction. So suppose there is an SRTM $M$ that computes $\text{dis}_{r}$. If we give $M$ only inputs of the form $n$, where $n \in \mathbb{N}$, then we know that $M$ always takes the same (non-deterministic) transitions, with the same weights $r \in R' \cup \{n\}$, where $R' \subset \mathbb{N}, |R'| < \infty$. 

It follows that there is a polynomial $p$ over $\mathbb{N}$ of the form 
$p(n) = \sum_{i=0}^k c_i{\cdot}n^i$ (where $c_i \in \mathbb{N}$, for all $i\geq0$) such that for these inputs, $\lVert M \rVert = p$. As $p(2) = 3$, it follows that  $c_i = 0$ for all $i \geq 2$, i.e., $p(n) = c_1{\cdot}n+c_0$. Furthermore, as $p(0) = 0$, it follows that $c_0=0$, and as $p(1) = 1$ that $c_1=1$. Hence $p(n) = n$, which contradicts the assertion that $p(2)=3$.
\end{proof}

\setcounter{theorem}{17}
\begin{theorem}[\NPoldinf{$\mathcal{R}$} versus \NPnewinf{$\mathcal{R}$}]\label{comparison}
For any commutative semiring $\mathcal{R}$, 
\[
\karp(\NPoldinf{\mathcal{R}}) = \NPnewinf{\mathcal{R}}^{\limrec}.
\]
\end{theorem}
\begin{proof}[Proof (sketch)]($\supseteq$)
For the $\supseteq$ inclusion, first note that \NPnewinf{$\mathcal{R}$} is closed under Karp s-reductions. This follows from Lemma 17, which is indeed sound for the new model. The same holds for \NPnewinf{$\mathcal{R}$}$^{\limrec}$. Thus, it suffices to show that \NPoldinf{$\mathcal{R}$} is contained in \NPnewinf{$\mathcal{R}$}$^{\limrec}$. 

For this, we show how to transform a transition function 
\[
\delta_{old} \subseteq \left(Q \times (\Sigma \cup R) \right)\times \left(Q \times (\Sigma \cup R)\right) \times \{-1,1\} \times R
\]
into an equivalent transition function 
\[
\delta_{new} \subseteq \left(Q \times \Sigma \right)\times \left(Q \times \Sigma\right) \times \{-1,1\} \times (R' \cup \{\mathcal{X}\} \cup \limrec).
\]

Due to the constraints on $\delta_{old}$, we can split $\delta_{old}$ into three types of transitions:
\begin{align*}
\delta_{1} &\subseteq \left(Q \times \Sigma \right)\times \left(Q \times \Sigma\right) \times \{-1,1\} \times R',\\
\delta_{2} &\subseteq Q \times Q \times \{-1,1\} \times R',\\
\delta_{3} &\subseteq Q \times Q \times \{-1,1\},
\end{align*}
where $\delta_{1}$ contains the transitions that we take, when there is no semiring value on the tape, and $\delta_{2}$ and $\delta_{3}$ contain the transitions, when there is a semiring value on the tape and we transition with a fixed value and the value under the head, respectively.

We construct $\delta_{new}$ as 
\begin{align*}
    \delta_{new} =& \phantom{\cup} \delta_{1}\\
                  & \cup \{(q,\mathcal{X}, q', \mathcal{X}, d, r) \mid (q, q', d, r) \in \delta_{2} \text{ and } (q, q', d) \not\in \delta_{3}\}\\
                  & \cup \{(q,\mathcal{X}, q', \mathcal{X}, d, \text{dis}_{R^*}) \mid (q, q', d) \in \delta_{3} \text{ and } R^* = \{r' \mid (q, q', d, r') \in \delta_{2}\}\}.
\end{align*}
Note that this transition function writes $\mathcal{X}$ if and only if it reads $\mathcal{X}$. Therefore, the configurations of the old and new model can be interchanged in the straightforward manner. That is, an old configuration $(q, s, n) \in Q \times (R\cup\Sigma)^{N} \times \mathbb{N}$ corresponds to the new configuration $(q, x(s), n) \in Q \times (\Sigma \times R)^{N} \times \mathbb{N}$, where $x(s)$ is as in \cref{def:initial} 
extended to infinite strings. Since our new configurations have $x_i^{\Sigma} = \mathcal{X}$ iff $x_i^{R}$ contains a semiring value that was contained in the input at that position, there is exactly one old configuration for every possible new configuration.

Let $(q, x, n)$ be a configuration of the old machine such that the current value under the tape is $\sigma^* \in \Sigma_{I}$. Then, 
\begin{align*}
\mathrm{val}(q, x, n) &= \bigoplus_{(q, x, n) \stackrel{r}{\rightarrow} (q', x', n')} r \otimes \mathrm{val}(q', x', n')\\
&= \bigoplus_{(q, \sigma^*, q', \sigma', d, r) \in \delta_1} r \otimes \mathrm{val}(q', x[n \mapsto \sigma'], n + d)\\
&= \bigoplus_{(q, \sigma^*, q', \sigma', d, r) \in \delta_{new}} r \otimes \mathrm{val}(q', x[n \mapsto \sigma'], n + d),
\end{align*}
where $x[n \mapsto \sigma']$ denotes the tape string, where $x[n \mapsto \sigma']_i = x_i$, if $i \neq n$ and $x[n \mapsto \sigma']_n = \sigma'$

Let $(q, x, n)$ be a configuration of the old machine such that the current value under the tape is $r^* \in R$. Then,
\begin{align*}
\mathrm{val}(q, x, n) &= \bigoplus_{(q, x, n) \stackrel{r}{\rightarrow} (q', x', n')} r \otimes \mathrm{val}(q', x', n')\\
&= \phantom{\oplus} \bigoplus_{(q, r^*, q', r^* d, r) \in \delta_{old}} r \otimes \mathrm{val}(q', x, n + d)\\
&= \phantom{\oplus} \bigoplus_{(q, q', d, r) \in \delta_2, \text{ s.t. } r \neq r^* \text{ or }(q,q',d) \not\in \delta_3} r \otimes \mathrm{val}(q', x, n + d)\\
&\phantom{\;=\;} \oplus \bigoplus_{(q, q', d) \in \delta_3} r^* \otimes \mathrm{val}(q', x, n + d)\\
&= \phantom{\oplus} \bigoplus_{(q, q', d, r) \in \delta_2, \text{ s.t. } (q,q',d) \not\in \delta_3} r \otimes \mathrm{val}(q', x, n + d)\\
&\phantom{\;=\;} \oplus \bigoplus_{(q, q', d, r) \in \delta_2, \text{ s.t. } (q,q',d) \in \delta_3 \text{ and } r \neq r^*} r \otimes \mathrm{val}(q', x, n + d)\\
&\phantom{\;=\;} \oplus \bigoplus_{(q, q', d) \in \delta_3} r^* \otimes \mathrm{val}(q', x, n + d)\\
&= \bigoplus_{q' \in \Sigma, d \in \{-1, 1\}}\left(\begin{array}{l}
     \phantom{\oplus} \bigoplus_{(q, q', d, r) \in \delta_2, \text{ s.t. } (q,q',d) \not\in \delta_3} r\\
\oplus \bigoplus_{(q, q', d, r) \in \delta_2, \text{ s.t. } (q,q',d) \in \delta_3 \text{ and } r \neq r^*} r\\
\oplus \bigoplus_{(q, q', d) \in \delta_3} r^* 
\end{array}\right)\otimes \mathrm{val}(q', x, n + d)\\
&= \bigoplus_{q' \in \Sigma, d \in \{-1, 1\}}\left(\begin{array}{lc}
\bigoplus_{(q, q', d, r) \in \delta_2} r & \text{if } (q,q',d) \not\in \delta_3\\
r^* \oplus \bigoplus_{(q, q', d, r) \in \delta_2, \text{ s.t. } r \neq r^*} r & \text{if } (q,q',d) \in \delta_3 
\end{array}\right)\otimes \mathrm{val}(q', x, n + d)
\end{align*}
Now let 
\[
R^*(q', d) = \{r' \mid (q, q', d) \in \delta_{3} \text{ and } (q, q', d, r') \in \delta_{2}\}.
\]
Then,
\begin{align*}
val&(q, x, n) = \bigoplus_{q' \in \Sigma, d \in \{-1, 1\}}\left(\begin{array}{lc}
\bigoplus_{(q, q', d, r) \in \delta_2} r & \text{if } (q,q',d) \not\in \delta_3\\
r^* \oplus \bigoplus_{(q, q', d, r) \in \delta_2, \text{ s.t. } r \neq r^*} r & \text{if } (q,q',d) \in \delta_3 
\end{array}\right)\otimes \mathrm{val}(q', x, n + d) \\
&= \bigoplus_{q' \in \Sigma, d \in \{-1, 1\}}\left(\begin{array}{lc}
\bigoplus_{(q, q', d, r) \in \delta_2} r & \text{if } (q,q',d) \not\in \delta_3\\
r^* \oplus \bigoplus_{r \in R^*(q',d), \text{ s.t. } r \neq r^*} r & \text{if } (q,q',d) \in \delta_3 
\end{array}\right)\otimes \mathrm{val}(q', x, n + d) \\
&= \bigoplus_{q' \in \Sigma, d \in \{-1, 1\}}\left(\begin{array}{lc}
\bigoplus_{(q, q', d, r) \in \delta_2} r & \text{if } (q,q',d) \not\in \delta_3\\
\rec{R^*(q',d)}(r^*) & \text{if } (q,q',d) \in \delta_3 
\end{array}\right)\otimes \mathrm{val}(q', x, n + d) \\
&= \bigoplus_{q' \in \Sigma, d \in \{-1, 1\}}\left(\begin{array}{lc}
\bigoplus_{(q, \mathcal{X}, q', \mathcal{X,} d, r) \in \delta_{new}} r & \text{if } (q,q',d) \not\in \delta_3\\
\bigoplus_{(q, \mathcal{X}, q', \mathcal{X}, d, \rec{R^*}) \in \delta_{new}} \rec{R^*}(r^*) & \text{if } (q,q',d) \in \delta_3 
\end{array}\right)\otimes \mathrm{val}(q', x, n + d) \\
&= \bigoplus_{q' \in \Sigma, d \in \{-1, 1\}}\left(
\bigoplus_{(q, \mathcal{X}, q', \mathcal{X}, d, r) \in \delta_{new}} r \oplus
\bigoplus_{(q, \mathcal{X}, q', \mathcal{X}, d, \rec{R^*}) \in \delta_{new}} \rec{R^*}(r^*) 
\right)\otimes \mathrm{val}(q', x, n + d) \\
&= 
\bigoplus_{(q, \mathcal{X}, q', \mathcal{X}, d, r) \in \delta_{new}} r 
\otimes \mathrm{val}(q', x, n + d) \oplus
\bigoplus_{(q, \mathcal{X}, q', \mathcal{X}, d, \rec{R^*}) \in \delta_{new}} \rec{R^*}(r^*)
\otimes \mathrm{val}(q', x, n + d) \\
&= 
\bigoplus_{(q, \mathcal{X}, q', \mathcal{X}, d, r) \in \delta_{new}} r 
\otimes \mathrm{val}(q', x, n + d) \oplus
\bigoplus_{(q, \mathcal{X}, q', \mathcal{X}, d, \rec{R^*}) \in \delta_{new}} \rec{R^*}(r^*)
\otimes \mathrm{val}(q', x, n + d) \\
&= \bigoplus_{\tau \in \delta_{new}(q, \mathcal{X})} \mathrm{wt}(\tau, c) \otimes \mathrm{val}(\tau(c))
\end{align*}
Here the last equality follows from the fact that the tape contains $r^*$, and hence the remaining transitions that read a semiring value, do not apply.

\noindent ($\subseteq$). As for the $\subseteq$ inclusion, we first simplify our task using the following auxiliary lemma:
\setcounter{theorem}{28}
\begin{lemma}
Let $M$ be an SRTM over $\mathcal{R}$ with oracle access to $\limrec$. Then there exists an equivalent SRTM $M'$ with oracle access to $\limrec$ such that for all $q,q' \in Q, \sigma, \sigma' \in \Sigma, d \in \{-1, 1\}$ there is exactly one transition $(q, \sigma, q', \sigma', d, w) \in \delta$.
\end{lemma}
\begin{proof}[Proof (sketch)]
For each $(q, \sigma, q', \sigma', d, w) \in \delta$ introduce new states $q_{(q, \sigma, q', \sigma', d, w), 0}$ and $q_{(q, \sigma, q', \sigma', d, w), 1}$,  and we use transitions 
\begin{align*}
    &(q, \sigma, q_{(q, \sigma, q', \sigma', d, w), 0}, \sigma, -1, e_{\otimes}), & &\\ &(q_{(q, \sigma, q', \sigma', d, w), 0}, \sigma^*, q_{(q, \sigma, q', \sigma', d, w), 1}, \sigma^*, 1, e_{\otimes}) & & \text{ for all } \sigma^* \in \Sigma,\\
    &(q_{(q, \sigma, q', \sigma', d, w), 1}, \sigma, q', \sigma', d, w). & & \\[-2\baselineskip]
\end{align*}
\end{proof}
To prove our result, we apply the following Karp s-reduction. Given a string $x \in (\Sigma_{I} \cup R)^*$ and a new constant $\mathcal{A}$ that does not occur in $\Sigma_{I}$,we define the alphabet $(\Sigma_{I})_{old} = \{(\sigma, \mathcal{A}) \mid \sigma \in \Sigma_{I}\} \cup \{(\mathcal{X}, \mathcal{X})\}$ and construct the string $x' \in ((\Sigma_{I})_{old} \cup R)^*$, where 
\begin{align*}
    x'_{2i + 1} &= \left\{\begin{array}{cc}
        x_i & \text{if } x_i \in R \\
        e_{\oplus} & \text{otherwise}
    \end{array}\right., & & \text{for } i = 0,\dots, |x| - 1,\\
    x'_{2i} &=\left\{\begin{array}{cc}
        (\mathcal{X}, \mathcal{X}) & \text{if } x_i \in R \\
        (x_i, \mathcal{A}) & \text{otherwise}
    \end{array}\right., & & \text{for } i = 0,\dots, |x| - 1.
\end{align*}
Intuitively, we replace alphabet letters $x_i$ with
$(x_i,\mathcal{A})$  and semiring values $x_i$ with $(\mathcal{X}, \mathcal{X})$, and we insert 
after $x_i$ zero respectively $x_i$.
\begin{example}
Consider the input string
$s = a,55,4,b,a,c,3 \in (\Sigma_{I} \cup \mathbb{N})^{*},$ \text{ where } $\Sigma_{I} = \{a,b,c\}$.
Applying the s-reduction yields the string
\[
s' = (a,\mathcal{A}),e_{\oplus},
    (\mathcal{X}, \mathcal{X}), 55,
    (\mathcal{X}, \mathcal{X}), 4,
    (b, \mathcal{A}), e_{\oplus},
    (a, \mathcal{A}), e_{\oplus},
    (c, \mathcal{A}), e_{\oplus},
    (\mathcal{X}, \mathcal{X}), 3
\]
That is, $a$ is expanded to the two symbols $(a,\mathcal{A})$ and $e_{\oplus}$. Further, 55 is expanded to the two symbols $(\mathcal{X}, \mathcal{X})$ and $55$, etc.
\end{example}

Then, given an SRTM $M=(R', Q, \Sigma_{I}, \Sigma, \iota, \mathcal{X}, \sqcup, \delta)$ over $\mathcal{R}$ with oracle access to $\limrec$, we construct a corresponding machine in the model of \cite{DBLP:journals/jair/EiterK23} as $$M_{old} = (R'_{old}, Q_{old}, (\Sigma_I)_{old}, \Sigma_{old}, \iota_{old}, \sqcup_{old}, \delta_{old})$$ over $\mathcal{R}$, where
\begin{align*}
    \Sigma_{old} &= \{(\sigma, \mathcal{A}) \mid \sigma \in \Sigma\} \cup \{(\sigma, \mathcal{X}) \mid \sigma \in \Sigma\} \cup \{(\mathcal{X}, \mathcal{X})\}\\
    \sqcup_{old} &= (\sqcup, \mathcal{A})\\
    \iota_{old} &= read(\iota)
\end{align*}
We construct $\delta_{old}$ as follows:

For each $(q, \sigma, q', \sigma', d, r') \in \delta$, we add
\begin{align*}
    (read(q), (\sigma, \mathcal{A}), step(q', d), (\sigma', \mathcal{A}), d, r')\\
    (read(q), (\sigma, \mathcal{X}), step(q', d), (\sigma', \mathcal{X}), d, r')\\
    (step(q', d), r, read(q'), r, d, e_{\otimes})  \text{ for } r \in R\\
    (step(q', d), \sqcup_{old}, read(q'), \sqcup_{old}, d, e_{\otimes})
\end{align*}
to transition with the correct weight and then skip the location, where the semiring value is stored.

For each $(q, \sigma, q', \sigma', d, \mathcal{X}) \in \delta$, we add
\begin{align*}
    (read(q), (\sigma, \mathcal{X}), weight(q, q', d, \mathcal{X}), (\sigma', \mathcal{X}), 1, e_{\otimes})\\
    weight(q, q', d, \mathcal{X}), r, twostep(q', d), r, -1, r) \text{ for } r \in R\\
    twostep(q', d), \sigma^*, step(q', d), \sigma^*, d, e_{\otimes}) \text{ for } \sigma^* \in \Sigma_{old}\\
    step(q', d), \sigma^*, read(q'), \sigma^*, d, e_{\otimes}) \text{ for } \sigma^* \in \Sigma_{old}
\end{align*}
to transition with the correct weight, go back, and take two steps into the specified direction, skipping the location, where the semiring value is stored.

For each $(q, \sigma, q', \sigma', d, rec_{R^*}) \in \delta$, we add
\begin{align*}
    (read(q), (\sigma, \mathcal{X}), weight(q, q', d, \rec{R^*}), (\sigma', \mathcal{X}), 1, e_{\otimes})\\
    (read(q), (\sigma, \mathcal{A}), weight(q, q', d, \rec{R^*}), (\sigma', \mathcal{A}), 1, e_{\otimes})\\
    (weight(q, q', d, \rec{R^*}), r, twostep(q'), r, -1, r) \text{ for } r \in R\\
    (weight(q, q', d, \rec{R^*}), r, twostep(q'), r, -1, r') \text{ for } r \in R, r' \in R^*\\
    (weight(q, q', d, \rec{R^*}), \sqcup_{old}, twostep(q'), \sqcup_{old}, -1, r') \text{ for } r' \in R^*\\
    twostep(q', d), \sigma^*, step(q', d), \sigma^*, d, e_{\otimes}) \text{ for } \sigma^* \in \Sigma_{old}\\
    step(q', d), \sigma^*, read(q'), \sigma^*, d, e_{\otimes}) \text{ for } \sigma^* \in \Sigma_{old}
\end{align*}
to transition with the correct weight, go back, and take two steps into the specified direction, skipping the location, where the semiring value is stored.

The states are the ones mention in $\delta_{old}$ and $R'_{old}$ consists of the fixed weights mentioned in $\delta_{old}$.

We can apply the same line of reasoning to show that the values of the configurations are the same. Here, we take multiple steps in the old model, where the new machine takes only one but (i) after these steps we arrive in the corresponding configuration, (ii) even though we introduce non-determinism, the configuration after the steps is always the same, and (iii) the transition weights are either the same or, in the $\rec{R*}$ case, sum up to the value of $\rec{R^*}(r)$ for each $r \in R$.
\end{proof}

\setcounter{theorem}{20}
\begin{proposition}
For every semiring $\mathcal{R}$, there exists a function $f \colon R^* \longrightarrow R$ that is \NPnewinf{$\mathcal{R}$}-complete w.r.t.\ Karp s-reductions.
\end{proposition}
\begin{proof}[Proof (sketch)]
Consider an \NPnewinf{$\mathcal{R}$}-complete problem computed by function $f_{c}$. We may assume that $\Sigma_{I} = \{t,f\}$.

Next, consider the following Karp s-reduction $\kappa\colon (\Sigma_I \cup R)^*\longrightarrow R^*$, such that $\kappa(s_0, \dots, s_{n - 1}) = r^{t}, r^{f}, r^{\mathcal{X}}, r^{\mathcal{R}},$ where
\begin{align*}
  &&  r^{t} &= r^{t}_0, \dots, r^{t}_{n - 1}, &  \text{s.t. } r^{t}_i &= \left\{\begin{array}{cc}
        e_{\otimes} & s_i = t, \\
        e_{\oplus} & \text{otherwise},
    \end{array}\right. \text{ for all } 0\leq i<n,  \\
&&   r^{f} &= r^{f}_0, \dots, r^{f}_{n - 1}, & \text{s.t. } r^{f}_i &= \left\{\begin{array}{cc}
        e_{\otimes} & s_i = f, \\
        e_{\oplus} & \text{otherwise},
    \end{array}\right. \text{ for all } 0\leq i<n,  \\
 & &  r^{\mathcal{X}} &= r^{\mathcal{X}}_0, \dots, r^{\mathcal{X}}_{n - 1}, & \text{s.t. } r^{\mathcal{X}}_i &= \left\{\begin{array}{cc}
        e_{\otimes} & s_i \in R, \\
        e_{\oplus} & \text{otherwise},
    \end{array}\right. \text{ for all } 0\leq i<n,  \\
 & &  r^{\mathcal{R}} &= r^{\mathcal{R}}_0, \dots, r^{\mathcal{R}}_{n - 1}, &\text{s.t. } r^{\mathcal{R}}_i &= \left\{\begin{array}{cc}
        s_i & s_i \in R, \\
        e_{\oplus} & \text{otherwise},
    \end{array}\right.  \text{ for all } 0\leq i<n. 
\end{align*}
That is, we replace $s$ by four strings of the same length, where the first one contains ones  ($e_\otimes$) in the positions where $s$ contains $t$'s and zeros everywhere else, the second contains ones in the positions, where $s$ contains $f$'s and zeros everywhere else, the third contains ones in the positions, where $s$ contains a semiring value, and the fourth contains the semiring values that $s$ has in the positions where it has semiring values and zeros everywhere else.

Next, we construct an SRTM $M$, with $\Sigma_{I} = \emptyset$, such that $M(\kappa(s)) = f_{c}(s)$. Then, the function defined by $M$ is \NPnewinf{$\mathcal{R}$}-complete and defined over $R^*$, as required.

We define $M$ according to the following algorithm:
\begin{enumerate}
    \item Check whether the length of the input $x$ is of the form $4n$. If not, halt with weight $e_{\oplus}$.
    \item Guess a string $s \in \{t,f,\mathcal{X}\}^n$. 
    \item For $i \in \{0, \dots, n\}$:
    \begin{enumerate}
        \item If $s_i = t$, go to position $i$; 
        \item If $s_i = f$, go to position $n + i$; 
        \item If $s_i = \mathcal{X}$, go to position $2n + i$. 
        transition with the weight in the current position.
    \end{enumerate}
    \item Finally, simulate  the function $f_{c}$ on $s^{\mathcal{R}}$, where $s^{\mathcal{R}}_i = x_{3n + i}$, if $s_i = \mathcal{X}$ and $s^{\mathcal{R}}_i = s_i$, otherwise.
\end{enumerate}
To verify that $M(\kappa(s)) = f(s)$, we note that for $\kappa(s)$ there exists exactly one guess $s^* \in \{t,f,\mathcal{X}\}^n$ such that the transitions in 3. all have weight $e_{\otimes}$. All other guesses lead to at least one transition with weight $e_{\oplus}$, such that the following computation does not contribute to the final result.

Since we simulate $f_{c}$ on the guessed input, for which  $(s^*)^{\mathcal{R}} = s$ holds, the final result is equal to $f_{c}(s)$ as desired.
\end{proof}

\setcounter{theorem}{21}
\begin{proposition}
There exists a function in $\NPnewinf{\mathcal{R}}$ that is not in $\NPfinite{\mathcal{R}}$.
\end{proposition}
\begin{proof}
It was previously shown  that for semirings that are not finitely generated, the identity function cannot be expressed in \NPfinite{$\mathcal{R}$}, cf. \cite{DBLP:conf/mfcs/BadiaDNP24}. 
More in detail, the identity function $f_{id}$ returns, given an input $x$, the  value $x_0$ in the first position of $x$ if it is a semiring value and zero ($e_\oplus)$ otherwise. 
The SRTM model we introduced can express $f_{id}$, as evidenced by the machine
$M = (\emptyset, \{\iota, f\}, \emptyset, \{\mathcal{X}, \sqcup\}, \iota, \sqcup, \delta)$, where 
\[
\delta = \{(\iota, \mathcal{X}, f, \mathcal{X}, 1, \mathcal{X})\}.
\]
If we have a semiring value as the first value in an input, we transition with its weight and stop. Thus, in this case, the output is the first value in the input.
\end{proof}

\subsection{Proofs adapted to the new machine model}
\label{appendix:fixed_proofs}
\setcounter{theorem}{18}
\begin{lemma}
Let $\mathcal{R}$ be a semiring and $f_i\colon (\Sigma \cup R)^{*} \longrightarrow R$, for $i =1,2$. If $f_2 \in \NPnewinf{\mathcal{R}}$ and $f_1$ is Karp s-reducible to $f_2$, then $f_1 \in \NPnewinf{\mathcal{R}}$.
\end{lemma}
\begin{proof}
Let $T$ be the reduction function and $R' = \{r_1, \dots, r_k\}$ the associated set of semiring constants. Let $M=(R_M', Q, \Sigma_{I}, \Sigma, \iota, \mathcal{X}, \sqcup, \delta)$ be an SRTM that solves $f_2$. We construct an SRTM $M'=(R_M' \cup R', Q', \Sigma_{I}, \Sigma', \iota', \mathcal{X}, \sqcup, \delta')$ that solves $f_1$. Note that we give $M'$ access to both $R_M'$, the set of semiring constants from $M$, and $R'$, the set of semiring constants, fixed in the reduction. 

While SRTMs cannot distinguish different semiring values, they can check whether there is a semiring value on the tape initially in a given position. Thus, given some input $x \in (\Sigma_{I} \cup R)^*$ we can obtain in polynomial time a string $x'\in s_{S,V,E}(\Sigma_{I})$ from $x$ by using $SV^{n+k}E$ for the $n^{\text{th}}$ semiring value occurrence in $x$ (from left to right). E.g.\ the string $x = x_1x_2x_3r_1x_4x_5r_2 \in (\Sigma_{I} \cup R)^{*}$, where $x_1, \dots, x_4 \in \Sigma_{I}$ and $r_1, r_2 \in R$, leads to the string $x' = x_1x_2x_3SVEx_4x_5SVVE \in s_{S,V,E}(\Sigma_{I})$ when $k = 0$. Then, we can apply $T$ to $x'$ in polynomial time. Finally, we can execute the algorithm for $f_2$ on $T(x')$. Here, we only need polynomial extra time to look up the $(n-k)^{\text{th}}$ semiring value in the original input when we need to make a transition with a semiring value under the head and read $SV^{n}E$ for $n > k$. For $n \leq k$ we cannot transition with a weight under the head, since $r_n$ does not necessarily occur in $x$. However, since we added $R'$ to the set of constants that $M'$ has access to, we can simply specify a transition with weight $r_n$ directly.
\end{proof}

Note that this proof is essentially the same as the previous proof, however, with the new machine model, it is actually always possible to look up the $(n-k)^{\text{th}}$ semiring value on the tape.

\begin{theorem}
\textsc{SAT}($\mathcal{R}$) is \NPnewinf{$\mathcal{R}$}-complete w.r.t.\ Karp s-reductions for every commutative semiring $\mathcal{R}$.
\end{theorem}

\begin{proof}
Containment follows from a simple algorithm that non-deterministically visits all subformulas of $\fplus$ formulas and deterministically/sequentially visits all subformulas of $\fmal$ formulas. For precise algorithm and containment proof see Algorithm 1 and the proof sketch of Theorem 18 in~\cite{DBLP:journals/jair/EiterK23}.

For hardness we generalize the Cook-Levin Theorem. So let $M=(R', Q, \Sigma_{I}, \Sigma, \iota, \mathcal{X}, \sqcup, \delta)$ be a polynomial time SRTM and $x \in (\Sigma \cup R)^*$ be the input for which we want to compute the output of $M$. 


We define the following propositional variables:
\begin{itemize}
    \item $T_{i,j,k}$, which is true if tape cell $i$ contains symbol $j$ at step $k$ of the computation.
    \item $H_{i,k}$, which is true if the $M$'s read/write head is at tape cell $i$ at step $k$ of the computation.
    \item $Q_{q,k}$, which is true if $M$ is in state $q$ at step $k$ of the computation. 
\end{itemize}
Furthermore, we need the following surrogates for semiring values
\begin{itemize}
    \item $\text{surr}_{i^{\text{th}}}$, which is the surrogate for the semiring value at the $i$\textsuperscript{th} position of the input $x$. Note that there may not be a semiring value at every index, in this case we assign $\text{surr}_{i^{\text{th}}}$ a fixed surrogate equal to $e_{\oplus}$, and
    \item $\text{surr}_{r_1}, \dots, \text{surr}_{r_m}$, which are respectively the surrogates for the constant semiring values $r_1, \dots, r_m$ that $M$ has access to. 
    \item These surrogates are all distinct.
\end{itemize}
Since $M$ is a polynomial time SRTM, we can assume the existence of a polynomial $p$ such that $p(n)$ bounds the number of transitions of $M$ on any input of length $n$.

Given a finite set $I$ and a family $(\beta_i)_{i \in I}$ of weighted formulas, we use the following shorthand:  
\[
\bigfplus_{i \in I} \beta_{i} = \left\{\begin{array}{cc}
    \szero & \text{ if } I = \emptyset \\
    \beta_{i^*} \fplus \bigfplus_{i \in I \setminus \{i^*\}} \beta_i & \text{ if } i^* \in I
\end{array}\right..
\]
Note that this is well defined, since $I$ is finite and addition is commutative and associative.

We define a weighted QBF $\bigfplus \mathbf{T}\; \bigfplus \mathbf{H} \;\bigfplus \mathbf{Q} \;\alpha$, where $\bigfplus \mathbf{T}$, $\bigfplus \mathbf{H}$, and $\bigfplus \mathbf{Q}$ correspond to the (sum) quantification of all variables $T_{i,j,k}$, $H_{i,k}$, and $Q_{q,k}$, respectively, and $\alpha$ is the product (i.e.\ connected with $\fmal$) of the following subformulas
\begin{enumerate}
    \item  $T_{i,j,0}$ \\
    Variable ranges: $0 \leq i < n$\\
    For each tape cell $i$ that initially contains symbol $j \in \Sigma_I$ or $j = \mathcal{X}$ when it contains a semiring value.  
    \item  $T_{i,\sqcup,0}$ \\
    Variable ranges: $-p(n) \leq i < 0$ or $n \leq i \leq p(n)$\\
    Each tape cell $i$ outside of the ones that contain the input contains $\sqcup$. 
    \item $Q_{\iota,0}$ \\
    The initial state of $M$ is $\iota$.
    \item $H_{0,0}$ \\ 
    The initial position of the head is $0$.
    \item $\neg T_{i,j,k}\fplus T_{i,j,k} \fmal \neg T_{i,j',k}$ \\ 
    Variable ranges: $-p(n) \leq i \leq p(n), j \in \Sigma, 0 \leq k \leq p(n)$\\
    There is at most one symbol per tape cell.
    \item $\bigfplus_{j\in \Sigma}T_{i,j,k}$ \\
    Variable ranges: $-p(n) \leq i \leq p(n), 0 \leq k \leq p(n)$\\
    There is at least one symbol per tape cell.
    \item $\neg T_{i,j,k} \fplus T_{i,j,k} \fmal \neg T_{i,j',k+1} \fplus T_{i,j,k} \fmal T_{i,j',k+1} \fmal H_{i,k}$ \\ 
    Variable ranges: $-p(n) \leq i \leq p(n),  j \neq j' \in \Sigma, 0 \leq k < p(n)$\\ Tape remains unchanged unless written.
    \item $\neg Q_{q,k} \fplus Q_{q,k} \fmal \neg Q_{q',k}$ \\ 
    Variable ranges: $q \neq q' \in Q, 0 \leq k \leq p(n)$\\
    There is at most one state at a time.
    \item $\neg H_{i,k} \fplus H_{i,k} \fmal \neg H_{i',k}$ \\ 
    Variable ranges: $i \neq i', -p(n) \leq i \leq p(n), -p(n) \leq i' \leq p(n), 0 \leq k \leq p(n)$\\
     There is at most one head position at a time.
    \item $\begin{array}{l}
        \neg H_{i,k}\fplus H_{i,k} \fmal \neg Q_{q,k} \fplus H_{i,k} \fmal Q_{q,k} \fmal \neg T_{i,\sigma ,k} \fplus \\
        H_{i,k} \fmal Q_{q,k} \fmal T_{i,\sigma ,k}\fmal \Sigma_{((q,\sigma ),(q',\sigma'),d, r)\in \delta} H_{i+d, k+1} \fmal Q_{q', k+1}\fmal T_{i, \sigma ', k+1} \fmal \text{surr}_{r}
    \end{array}$ \\ 
    Variable ranges: $-p(n) \leq i \leq p(n), 0 \leq k < p(n)$ and $q \in Q$, $\sigma \in \Sigma$ s.t. there exist $q', 
    \sigma', d, r$ s.t. $((q,\sigma ),(q',\sigma'),d, r)\in \delta$. \\
    Possible weighted transitions with a fixed weight at computation step $k$ when head is at position $i$.\\
    \item $\begin{array}{l}
        \neg H_{i,k}\fplus H_{i,k} \fmal \neg Q_{q,k} \fplus H_{i,k} \fmal Q_{q,k} \fmal \neg T_{i,\sigma ,k} \fplus \\
        H_{i,k} \fmal Q_{q,k} \fmal T_{i,\sigma ,k}\fmal \Sigma_{((q,\sigma ),(q',\sigma'),d, \mathcal{X})\in \delta} H_{i+d, k+1} \fmal Q_{q', k+1}\fmal T_{i, \sigma ', k+1} \fmal \text{surr}_{i^{\text{th}}}
    \end{array}$ \\ 
    Variable ranges: $-p(n) \leq i \leq p(n), 0 \leq k < p(n)$ and $q \in Q$, $\sigma \in \Sigma$ s.t. there exist $q', 
    \sigma', d$ s.t. $((q,\sigma ),(q',\sigma'),d, \mathcal{X})\in \delta$. \\
    Possible weighted transitions with weight from the tape at computation step $k$ when head is at position $i$.\\
    \item $\begin{array}{c}
        \phantom{\fplus}\neg H_{i,k}\fplus H_{i,k} \fmal \neg Q_{q,k} \fplus H_{i,k} \fmal Q_{q,k} \fmal \neg T_{i,\sigma ,k} \\
        \fplus
        H_{i,k} \fmal Q_{q,k} \fmal T_{i,\sigma ,k}\fmal  H_{i, k+1} \fmal Q_{q, k+1}\fmal T_{i, \sigma, k+1}
    \end{array}$ \\ 
    Variable ranges: $-p(n) \leq i \leq p(n), 0 \leq k < p(n)$ and $q \in Q$, $\sigma \in \Sigma$ s.t. there do not exist $q', 
    \sigma', d, r$ s.t. $((q,\sigma ),(q',\sigma'),d, r)\in \delta$. \\
    The machine computation has ended. This is included so that when the machine has reached a final state it stays the same until $k$ reaches $p(n)$.
\end{enumerate}
In order to prove correctness, i.e., that the value of $M$ on $x$ is equal to $\llbracket \bigfplus \mathbf{T}\; \bigfplus \mathbf{H} \;\bigfplus \mathbf{Q} \;\alpha \rrbracket_{\mathcal{R}}(\emptyset)$, we prove two claims.
\begin{enumerate}
    \item For each interpretation $\mathcal{I}$ of the variables $T_{i,j,k}, H_{i,k}, Q_{q,k}$ for $-p(n) \leq i \leq p(n), 0 \leq k < p(n)$ and $q \in Q$ such that $\mathcal{I}$ does not correspond to a computation path of $M$ on $x$, it holds that $\llbracket \alpha \rrbracket_{\mathcal{R}}(\mathcal{I}) = \szero$.
    \item For each interpretation $\mathcal{I}$ of the variables $T_{i,j,k}, H_{i,k}, Q_{q,k}$ for $-p(n) \leq i \leq p(n), 0 \leq k < p(n)$ and $q \in Q$ such that $\mathcal{I}$ corresponds to a computation path $\pi$ of $M$ on $x$ along configurations $c_1^{\pi} \stackrel{r^{(\pi_1)}}{\rightarrow} \dots \stackrel{r^{(\pi_{n(\pi)-1})}}{\rightarrow} c_{n(\pi)}^{\pi}$ it holds that 
    \[
    \llbracket \alpha \rrbracket_{\mathcal{R}}(\mathcal{I}) = \bigfplus_{\pi', \text{ s.t.} c_i^{\pi} = c_{i}^{\pi'}}r^{(\pi'_1)} \fmal \dots \fmal r^{(\pi'_{n(\pi')-1})}.
    \]
\end{enumerate}
If both claims hold, then it follows that $\llbracket \bigfplus \mathbf{T}\; \bigfplus \mathbf{H} \;\bigfplus \mathbf{Q} \;\alpha \rrbracket_{\mathcal{R}}(\emptyset)$ is equal to the sum of $\llbracket \alpha \rrbracket_{\mathcal{R}}(\mathcal{I})$ over all interpretations $\mathcal{I}$ such that $\mathcal{I}$ corresponds to a computation path. For each of them, we know that the weight of the path is the product of the weights of the taken transitions, according to (ii). Since the value of $M$ on $x$ is equal to the sum of the weights of the paths, this implies correctness of the reduction.

We proceed to prove the claims. 
For this, we first generally show that the added subformulas 1.\ to 12.\ enforce their given purpose.

For 1.\ to 4.\ this is clear: In order for $\llbracket \alpha \rrbracket_{\mathcal{R}}(\mathcal{I})$ to be unequal to $\szero$, the variables need to be included in the interpretation.

5.\ and 6.\ together ensure that at each time step there is exactly one symbol in each tape cell. So assume that $T_{i,j,k}$ and $T_{i,j',k}$ are in $\mathcal{I}$. Then 
\[
\llbracket \neg T_{i,j,k}\fplus T_{i,j,k} \fmal \neg T_{i,j',k} \rrbracket_{\mathcal{R}}(\mathcal{I}) = \szero \fplus \szero = \szero,
\]
and so $\llbracket \alpha \rrbracket_{\mathcal{R}}(\mathcal{I}) = \szero$. 
Assume alternatively that there are $i,k$ such that for no $j$ the variable $T_{i,j,k}$ is in $\mathcal{I}$. Then
\[
\llbracket \bigfplus_{j\in \Sigma }T_{i,j,k} \rrbracket_{\mathcal{R}}(\mathcal{I}) = \szero.
\]

7.\ ensures that the tape remains unchanged unless written, i.e., unless the head is at position $i$ at step $k$ the value of the tape cell $i$ stays the same. So assume that $H_{i,k}$ is not in $\mathcal{I}$ but $T_{i,j,k}$ and $T_{i,j',k}$ for $j \neq j'$ are in $\mathcal{I}$. Then
\[
\llbracket \neg T_{i,j,k} \fplus T_{i,j,k} \fmal \neg T_{i,j',k+1} \fplus T_{i,j,k} \fmal T_{i,j',k+1} \fmal H_{i,k} \rrbracket_{\mathcal{R}}(\mathcal{I}) = \szero \fplus (\sone \fmal \szero) \fplus (\sone \fmal \sone \fmal \szero) = \szero.
\]

8.\ ensures that there is at most one state at a time by analogous reasoning to 5.

9.\ ensures that there is at most one head position at a time by analogous reasoning to 5.

10. and 11.\ model the possible transitions at computation step $k$ when the head is at position $i$ including their respective weights \emph{if there is such a transition} for the given state and tape cell entry. Note for this that the formulas are the same except for the weight, expressed by the surrogate Otherwise, this subformula is not added but the one in 12.\ is. So assume that $H_{i,k}, Q_{q,k}, T_{i,\sigma,k} \in \mathcal{I}$. Then the value of the subformula 10.\, when replacing the surrogate $\text{surr}_{r}$ by the corresponding semiring value $r$, is 
\begin{align*}
    &\llbracket \neg H_{i,k}\fplus H_{i,k} \fmal \neg Q_{q,k} \fplus H_{i,k} \fmal Q_{q,k} \fmal \neg T_{i,\sigma ,k} \rrbracket_{\mathcal{R}}(\mathcal{I}) \fplus \\
    &\llbracket H_{i,k} \fmal Q_{q,k} \fmal T_{i,\sigma ,k}\fmal \Sigma_{((q,\sigma ),(q',\sigma'),d, r)\in \delta'} H_{i+d, k+1} \fmal Q_{q', k+1}\fmal T_{i, \sigma ', k+1} \fmal r \rrbracket_{\mathcal{R}}(\mathcal{I})\\
    = \phantom{\fplus} & \szero \fplus \sone \fmal \llbracket \Sigma_{((q,\sigma ),(q',\sigma'),d, r)\in \delta'} H_{i+d, k+1} \fmal Q_{q', k+1}\fmal T_{i, \sigma ', k+1} \fmal r \rrbracket_{\mathcal{R}}(\mathcal{I})\\
    = \phantom{\fplus} & \llbracket \Sigma_{((q,\sigma ),(q',\sigma'),d, r)\in \delta'} H_{i+d, k+1} \fmal Q_{q', k+1}\fmal T_{i, \sigma ', k+1} \fmal r \rrbracket_{\mathcal{R}}(\mathcal{I})
\end{align*}
The analogous holds for the formula of 11. when replacing the surrogate  $\text{surr}_{i^{\text{th}}}$ by the corresponding semiring value $r_{i^\text{th}}$ at position $i$ of the input.

This means that the expression evaluates to the sum of all weights of the transitions we take. Note that in order for two transitions $((q,\sigma ),(q_1',\sigma_1'),d_1, r_1),$ $((q,\sigma ),(q_2',\sigma_2'),d_2, r_2)$ to be different at least one of $q_1' \neq q_2'$, $\sigma_1' \neq \sigma_2'$, $d_1 \neq d_2$ or $r_1 \neq r_2$ needs to hold. If $q_1' \neq q_2'$, $\sigma_1' \neq \sigma_2'$ or $d_1 \neq d_2$ then one of 5., 8., or 9.\ is falsified. Thus, in this case, we can take multiple transitions if they differ in the weights only. On the other hand, we must take at least one transition, since otherwise the whole sum evaluates to $\szero$. It follows that we transition to exactly one new configuration to obtain a non-zero value for $\alpha$. In that case, we have $H_{i+d, k+1}, Q_{q', k+1}, T_{i, \sigma ', k+1} \in \mathcal{I}$ for the corresponding transition(s) $((q,\sigma ),(q',\sigma'),d, r) \in \delta$ and 
\begin{align*}
&\llbracket \Sigma_{((q,\sigma ),(q',\sigma'),d, r)\in \delta'} H_{i+d, k+1} \fmal Q_{q', k+1}\fmal T_{i, \sigma ', k+1} \fmal r \rrbracket_{\mathcal{R}}(\mathcal{I}) \\
\fplus &\llbracket \Sigma_{((q,\sigma ),(q',\sigma'),d, \mathcal{X})\in \delta'} H_{i+d, k+1} \fmal Q_{q', k+1}\fmal T_{i, \sigma ', k+1} \fmal r_{i^{\text{th}}} \rrbracket_{\mathcal{R}}(\mathcal{I}) \\
= &\bigfplus_{((q,\sigma ),(q',\sigma'),d, r)\in \delta} r \\
\fplus &\bigfplus_{((q,\sigma ),(q',\sigma'),d, \mathcal{X})\in \delta} r_{i^{\text{th}}}.
\end{align*}

12.\ models that if at computation step $k$ when the head is at position $i$ \emph{there is no possible transition} for the given state $q$ and tape cell entry $\sigma$, then the head position, state and tape cell entries stay the same.
Otherwise, this subformula is not added but one in 10.\ is. So assume that $H_{i,k}, Q_{q,k}, T_{i,\sigma,k} \in \mathcal{I}$. Then the value of the subformula is 
\begin{align*}
    &\llbracket \neg H_{i,k}\fplus H_{i,k} \fmal \neg Q_{q,k} \fplus H_{i,k} \fmal Q_{q,k} \fmal \neg T_{i,\sigma ,k} \rrbracket_{\mathcal{R}}(\mathcal{I}) \fplus \\
    &\llbracket H_{i,k} \fmal Q_{q,k} \fmal T_{i,\sigma ,k}\fmal  H_{i, k+1} \fmal Q_{q, k+1}\fmal T_{i, \sigma, k+1}  \rrbracket_{\mathcal{R}}(\mathcal{I})\\
    = \phantom{\fplus} & \szero \fplus \sone \fmal \llbracket H_{i, k+1} \fmal Q_{q, k+1}\fmal T_{i, \sigma, k+1} \rrbracket_{\mathcal{R}}(\mathcal{I})
\end{align*}
This means that the expression evaluates to $\sone$, if $H_{i,k+1}, Q_{q,k+1}, T_{i,\sigma,k+1} \in \mathcal{I}$ and evaluates to $\szero$, otherwise. Thus, the formula enforces the desired constraint: if we are in a configuration without further transitions we stay in it until the time limit $p(n)$ is reached. Notably, this does not influence the weight, since we always obtain a factor of $\sone$.

Putting things together, we see that 8.\ and 10./11./12.\ together ensure that there is exactly one state at each time. Similarly, 9.\ and 10./11./12.\ together ensure that there is exactly one head position at each time. This together with the constraints associated originally with 1.\ to 12.\ show that the desired claims and therefore also the theorem hold.
\end{proof}

\section{An application in database theory: data complexity of (Unions of) Conjunctive Queries}\label{appendix:UCQ}

Green at al.~\cite{GKV} introduced a semiring-based semantics for relational algebra queries.  Its main use case as described in the paper is to allow queries for the provenance of answers to standard queries.  The semantics is capable of expressing bag semantics, why-provenance and more.  Here intuitively, bag semantics tell us not only whether we can derive a query but also how often we can derive it.  Why-provenance tells us instead the different reasons why we  have  to  derive  a  query.   There  are  more  possibilities  but  generally,  the  semiring semantics allows us to obtain more, often quantitative, information about the derivations of a query.

For positive logic programs, i.e.\ datalog programs, their semantics over a commutative semiring $(R, \splus, \smal, \szero, \sone)$ is as follows: the label of a query result $q(\overline{x})$ is the sum (using $\splus$) of the labels of derivation trees for $q(\overline{x})$, where the label of a derivation tree is the product (using $\smal$) of the labels of the leaf nodes (i.e. extensional atoms). As the number of derivation trees may be countably infinite, Green et al. used $\omega$-continuous semirings, where countable sums have a value. Examples of such semirings are $\mathbb{N}_{\infty} = (\mathbb{N}\cup \{\infty\}, +, \cdot, 0, 1)$, the natural number semiring extended with infinity, $\mathbb{B}$, and other idempotent semirings.

\setcounter{theorem}{30} 
\begin{example}[Bag Semantics]
\label{ex:trans}

Consider the program
\begin{align*}
    r_1{:}\; q(X,Y) &\leftarrow r(X,Y) & 
    r_2{:}\; q(X,Y) &\leftarrow q(X,Z), q(Z,Y)
\end{align*}
over $\mathbb{N}_{\infty}$ and the semiring-weighted extensional database (edb) 
\[
\{(r(a,b),2), (r(b,c),3), (r(b,a),4)\}.
\]
Here, $\mathbb{N}_{\infty}$ corresponds to bag semantics, which means that the label of a query $q(\overline{x})$ is the number of derivations of $q(\overline{x})$. The labels of the atoms in the semiring-weighted edb thus intuitively mean that there are $2, 3$, and $4$ ways to derive $r(a,b), r(b,c)$, and $r(b,a)$, respectively.

It follows that the label of $q(b,c)$ under bag semantics is $3$ due to the derivation $q(b,c) \leftarrow r(b,c)$ from $r(b,c)$, which itself has $3$ derivations according to the edb, and the fact that this is the only derivation for $q(b,c)$. On the other hand, for $q(a,a)$ the label under bag semantics is $\infty$. This can be seen as follows: we can derive $q(a,a)$ from $r(a,b), r(b,a)$ in $2\,{\times}\,4$ ways according to their labels in the edb. But there exist infinitely many derivations of $q(a,a)$ from itself using rule $r_2$ instantiated as $q(a,a) \leftarrow q(a,a), q(a,a)$, leading to the label $\infty$.
\end{example}



The \noindent\textbf{Combined Complexity of Conjunctive Queries (CQs)} is the following problem:
Given a positive, single-rule datalog program 
\[
\Pi = \{ q \leftarrow r_1(\vec{Y_1}), \dots, r_n(\vec{Y_n}) \}
\]
and a semiring-weighted extensional database $D$ over a commutative semiring $\mathcal{R}$ as input, compute the label of $q$.

It was shown that combined complexity of CQs is \NPnewinf{$\mathcal{R}$}-complete~\cite{DBLP:journals/jair/EiterK23}:

\begin{theorem}[Theorem~26 of \cite{DBLP:journals/jair/EiterK23}]
\label{thm:harddataCQ}
Given a positive, single-rule datalog program 
\[
\Pi = \{ q \leftarrow r_1(\vec{Y_1}), \dots, r_n(\vec{Y_n}) \}
\]
and a semiring-weighted extensional database $D$ over a commutative semiring $\mathcal{R}$ as input, it is \NPnewinf{$\mathcal{R}$}-complete w.r.t.\ Karp s-reductions to compute the label of $q$.
\end{theorem}

In the following, we will show that the same holds true for unions of conjunctive queries with semiring semantics.

The \noindent\textbf{Combined Complexity of Unions of Conjunctive Queries (UCQs)} is the following problem:
Given a set of positive, single-rule datalog program 
\[
\{\Pi_1, \dots, \Pi_m\}
\]
and a semiring-weighted extensional database $D$ over a commutative semiring $\mathcal{R}$ as input, compute the sum of the labels of $q_1, \dots, q_n$.
\begin{lemma}[CQ vs UCQ]
    There is a Karp s-reduction from UCQ to CQ.
\end{lemma}
\begin{proof}[Proof (Sketch)]
    We add a new constant $\bot$ and next add to the extensional database $D$, the edb
    \begin{align*}
    &\{(r(\vec{\bot}), \sone) \mid r \text{ is a predicate in }D\}\\
    \cup &\{(\text{iff}(x, x'), \sone), (\text{iff}(\bot,\bot), \sone) \mid x, x' \text{ is a constant from }D\}\\
    \cup &\{(\text{exactlyOne}(x, \bot, \dots, \bot), \sone), \dots,(\text{exactlyOne}(\bot, \dots, \bot, x), \sone) \mid x \text{ is a constant from }D\}.
    \end{align*}
    Then the query
    \begin{align*}
    q \leftarrow &r_1^{(1)}(\vec{Y}_{1}^{(1)}), \dots, r_{n_1}^{(1)}(\vec{Y}_{n_1}^{(1)}), \\
    &\dots\\
    &r_1^{(m)}(\vec{Y}_{1}^{(m)}), \dots, r_{n_m}^{(m)}(\vec{Y}_{n_m}^{(m)}), \\
    &\text{iff}(Y_{1,1}^{(1)}, Y_{1,2}^{(1)}), \dots, \text{iff}(Y_{n_1,k_{n_1} - 1}^{(1)}, Y_{n_1,k_{n_1}}^{(1)}), \\
    &\dots\\
    &\text{iff}(Y_{1,1}^{(m)}, Y_{1,2}^{(m)}), \dots, \text{iff}(Y_{n_m,k_{n_m} - 1}^{(m)}, Y_{n_m,k_{n_m}}^{(m)}), \\
    &\text{exactlyOne}(Y_{1,1}^{(1)}, \dots, Y_{1,1}^{(m)})
    \end{align*}
    That is, we have all $m$ basic queries, with mutually different variables. Additionally, the ``iff'' lines ensure that either all variables from the $i$-th query are assigned $\bot$ or none are.
    Last but not least the ``exactlyOne'' predicate ensures that exactly one of the sets of variables has proper values assigned and all others have $\bot$ assigned.

    Thus, there is a bijection between the assignments to the variables that have a non-zero value and the tuples of queries $q_i$ and assignments $\sigma_i$ such that $q_i$ has a non-zero value for the variable assignment $\sigma_i$.

    Since $r(\vec{\bot})$ has value $\sone$, for each of the original predicates $r$, the values are preserved as desired.
\end{proof}

\begin{corollary}

\label{thm:harddataU}
The combined complexity of UCQs is \NPnewinf{$\mathcal{R}$}-complete w.r.t.\ Karp s-reductions.
\end{corollary}

\cref{thm:harddata} then follows from \cref{thm:harddataCQ} and \cref{thm:harddataU}.

\begin{figure}[p]
    \centering
    \begin{tikzpicture}[rotate=57.5,transform shape]
    \node (count) at (0, 1) {\#P-like};
    \node (modnp) at (3, 1) {$\textsc{Mod}_p\textsc{P}\cup\textsc{NP}$-like};
    \node (mod) at (6, 1) {$\textsc{Mod}_p\textsc{P}$-like};
    \node (sat) at (9, 1) {\textsc{NP}-like};
    \node (05a) at (0, 0) {$\mathbb{N}[(x_i)_{\infty}]$};
    \node (1b) at (3, 0) {$\mathbb{Z}_{p}\times\mathbb{N}_{\leq o}[(x_i)_{\infty}]$};
    \node (1c) at (6,-1) {$\mathbb{Z}_{p}[(x_i)_{\infty}]$};
    \node (1d) at (9,-1) {$\mathbb{N}_{\leq o}[(x_i)_{\infty}]$};
    \node (2d) at (9,-2) {$\mathbb{B}[(x_i)_{\infty}]$};
    \node (25a) at (0,-3) {$\mathbb{N}[(x_i)_{k}]$};
    \node (3b) at (3,-3) {$\mathbb{Z}_{p}\times\mathbb{N}_{\leq o}[(x_i)_{k}]$};
    \node (3c) at (6,-4) {$\mathbb{Z}_{p}[(x_i)_{k}]$};
    \node (3d) at (9,-4) {$\mathbb{B}[(x_i)_{k}]$};
    
    \node (4a) at (0,-5) {$\mathbb{Q}$};
    \node (45a) at (0,-6) {$\mathbb{Z}$};
    \node (5a) at (0,-7) {$\mathbb{N}$};
    
    \node (5b) at (3,-7) {$\mathbb{Z}_{p}\times\mathbb{N}_{\leq o}$};
    
    \node (5c) at (6,-8) {$\mathbb{Z}_{p}$};
    
    \node (4d1) at (8.5,-6) {$\mathcal{R}_{\max}$};
    \node (4d2) at (9.5,-6) {$\mathcal{R}_{\min}$};
    \node (5d) at (9,-8) {$\mathbb{B}$};
    
    \node (6d) at (6,-9) {$\mathbb{T}$};

    \draw[->] (05a) -- (1b);
    \draw[->] (1b) -- (1c);
    \draw[->] (1b) -- (1d);
    \draw[->] (05a) -- (25a);
    \draw[->] (25a) -- (4a);
    \draw[->] (1b) -- (3b);
    \draw[->] (25a) -- (3b);
    \draw[->] (1c) -- (3c);
    \draw[->] (1d) -- (2d);
    \draw[->] (2d) -- (3d);
    \draw[->] (3b) -- (3c);
    \draw[->] (3b) -- (3d);
    \draw[->] (4a) -- (45a);
    \draw[->] (45a) -- (5a);
    \draw[->] (3b) -- (5b);
    \draw[->] (5a) -- (5b);
    \draw[->] (3c) -- (5c);
    \draw[->] (5b) -- (5c);
    \draw[->] (3d) -- (4d1);
    \draw[->] (3d) -- (4d2);
    \draw[->] (4d1) -- (5d);
    \draw[->] (4d2) -- (5d);
    \draw[->] (5b) -- (5d);
    \draw[->] (5a) -- (6d);
    \draw[->] (5b) -- (6d);
    \draw[->] (5c) -- (6d);
    \draw[->] (5d) -- (6d);
    
    \draw [decoration={brace,amplitude=2mm},decorate,thick]
      (10,0) -- (10,-2.5)
      node [midway,right=4mm,align=left] {\small \textsc{FPSpace(poly)}-\\\small membership implies \\\small a collapse of \textsc{PH}};
    \draw[-, black, dotted] (-4,-2.5) -- (12,-2.5);
    
    \node (FPparrCP) at (-3,-3) {\small$\textsc{FP}^{\#P}_{\|}$-complete};
    \draw[-, black, densely dotted] (FPparrCP) -- (25a);
    
    \node (FPparrModNP) at (12,-3) {\small$\textsc{FP}^{\textsc{Mod}_p\textsc{P}\cup\textsc{NP}}_{\|}$-complete};
    \draw[-, black, densely dotted] (FPparrModNP) -- (3b);
    
    \node (FPparrMod) at (-3,-4) {\small$\textsc{FP}^{\textsc{Mod}_p\textsc{P}}_{\|}$-complete};
    \draw[-, black, densely dotted] (FPparrMod) -- (3c);
    
    \node (FPparrNP) at (12, -4) {\small$\textsc{FP}^{\textsc{NP}}_{\|}$-complete};
    \draw[-, black, densely dotted] (FPparrNP) -- (3d);
    
    \node (GapP) at (-3,-6) {\small\textsc{GapP}-complete};
    \draw[-, black, densely dotted] (GapP) -- (45a);
    
    \node (OptP) at (12, -6) {\small\textsc{OptP}-complete};
    \draw[-, black, densely dotted] (OptP) -- (4d1);
    
    \node (CP) at (-3,-7) {\small\#P-complete};
    \draw[-, black, densely dotted] (CP) -- (5a);
    
    \node (ModP) at (-3,-8) {\small$\textsc{Mod}_p\textsc{P}$-hard};
    \draw[-, black, densely dotted] (ModP) -- (5c);
    
    \node (NP) at (12, -8) {\small\textsc{NP}-complete};
    \draw[-, black, densely dotted] (NP) -- (5d);
    
    \node (trivial) at (-3, -9) {\small $\mathcal{O}(1)$};
    \draw[-, black, densely dotted] (trivial) -- (6d);

    \end{tikzpicture}
    \caption{Intuitive hardness relation between semirings. $\mathcal{R}_1$ harder than $\mathcal{R}_2$ indicated by arrows $\mathcal{R}_1 \rightarrow \mathcal{R}_2$. Relation of complexity classes $\mathcal{C}$ and semirings $\mathcal{R}$, indicated by dotted lines $\mathcal{C}$\protect\tikz[baseline]{\protect\draw[densely dotted] (0,0.5ex)--(0.3,0.5ex);}$\mathcal{R}$.} 
    \label{fig:subsumptions}
\end{figure}

\Cref{fig:subsumptions} shows the intuitive hardness relationship between semirings, as formalized by epimorphisms conforming to \cite{DBLP:journals/jair/EiterK23}. We see that, broadly speaking, there are four types of semirings associated with \#\textsc{P}, $\textsc{Mod}_p\textsc{P}$, \textsc{NP}, or a combination the last two. In a sense, the semiring of polynomials over countably many variables with natural number coefficients is structurally as hard as possible~\cite{DBLP:journals/jair/EiterK23}.

\end{document}